\documentclass{article}
\usepackage[utf8]{inputenc}
\usepackage[T1]{fontenc}
\usepackage[english]{babel}
\usepackage{amsmath,amsthm,amssymb,amsfonts}
\usepackage{xspace}
\usepackage{hyperref}
\usepackage{cleveref}
\usepackage{verbatim}
\usepackage{thm-restate}
\usepackage{tikz}
\usepackage{multirow,hhline}

\usepackage[hmargin=2.5cm,vmargin=3.8cm]{geometry}
\usepackage[textsize=footnotesize,color=green!40]{todonotes}

\usepackage{bm}

\newcommand{\FPT}{$\mathsf{FPT}$ }

\newcommand{\RBDS}{\textsc{Red-Blue Dominating Set}\xspace}
\newcommand{\CFDPO}{\textsc{Connected Fair Division}\xspace}
\newcommand{\CFDP}{\textsc{Incomplete Connected Fair Division}\xspace}
\usepackage{todonotes}
\newcommand{\CFD}{$\mathsf{ICFD}$\xspace}
\newcommand{\CFDO}{$\mathsf{CFD}$\xspace}
\newcommand{\PROP}{$\mathsf{PROP}$\xspace}
\newcommand{\EF}{$\mathsf{EF}$\xspace}
\newcommand{\EFO}{$\mathsf{EF1}$\xspace}
\newcommand{\EFX}{$\mathsf{EFX}$\xspace}
\newcommand{\I}{\text{I}}
\newcommand{\II}{\text{II}}
\newcommand{\set}[1]{\left\{#1\right\}}
\newcommand{\ksum}{$(k,M)$-\textsc{SUM}\xspace}

\newcommand{\NPoly}{\textsf{NP} $\subseteq$ \textsf{coNP}$/$\textsf{poly}\xspace}

\newtheorem{theorem}{Theorem}
\crefname{theorem}{theorem}{theorems}
\newtheorem{lemma}[theorem]{Lemma}
\crefname{lemma}{lemma}{lemmas}
\newtheorem{proposition}[theorem]{Proposition}

\crefname{proposition}{proposition}{propositions}
\crefname{result}{result}{results}
\newtheorem{corollary}[theorem]{Corollary}
\crefname{corollary}{corollary}{corollaries}

\crefname{fact}{fact}{facts}
\newtheorem{observation}[theorem]{Observation}
\crefname{observation}{observation}{observations}

\crefname{question}{question}{questions}

\crefname{claim}{claim}{claims}

\crefname{note}{note}{notes}

\crefname{conj}{conjecture}{conjectures}
\newtheorem{definition}[theorem]{Definition}
\crefname{definition}{definition}{definitions}
\newtheorem{remark}[theorem]{Remark}
\crefname{remark}{remark}{remarks}

\newtheorem{theorem1}{Theorem}
\newtheorem{RR}[theorem1]{Reduction Rule}

\newcounter{claimcounter}
\numberwithin{claimcounter}{lemma}
\newenvironment{clm}{\refstepcounter{claimcounter}{\medskip\noindent \textit{Claim} \theclaimcounter:}}{}

\newenvironment{proofofclaim}{%
    
  \proof}{\endproof}

\tikzstyle{noeud}=[circle,inner sep=2, minimum size =3 pt, line width = 1pt, draw=black, fill=white]

\title{Parameterized Complexity of Incomplete Connected Fair Division}
\author{
{Harmender Gahlawat} and {Meirav Zehavi}\\
\mbox{}\\
{\small  Ben-Gurion University of the Negev, Beersheba, Israel}\\}

\begin{document}

\maketitle

\begin{abstract}
\textit{Fair division} of resources among competing agents is a fundamental problem in computational social choice and economic game theory. It has been intensively studied on various kinds of items (\textit{divisible} and \textit{indivisible}) and under various notions of \textit{fairness}. We focus on \CFDPO (\CFDO), the variant of fair division on graphs, where the \textit{resources} are modeled as an \textit{item graph}. Here, each agent has to be assigned a connected subgraph of the item graph, and each item has to be assigned to some agent. 

We introduce a generalization of \CFDO, termed \textsc{Incomplete}  \CFDO (\CFD), where exactly $p$  vertices of the item graph should be assigned to the agents. This might be useful, in particular when the allocations are intended to be ``economical'' as well as fair. We consider four well-known notions of fairness: \PROP, \EF, \EFO, \EFX. First, we prove that \EF-\CFD, \EFO-\CFD, and \EFX-\CFD are W[1]-hard parameterized by $p$ plus the number of agents, even for graphs having constant \textit{vertex cover number} ($\mathsf{vcn}$). In contrast, we present a randomized \FPT algorithm for \PROP-\CFD parameterized only by $p$. Additionally, we prove both positive and negative results concerning the kernelization complexity of \CFD under all four fairness notions, parameterized by $p$, $\mathsf{vcn}$, and the total number of different valuations in the item graph ($\mathsf{val}$). 
\end{abstract}

\section{Introduction}
Allocating indivisible goods among competing agents in a ``fair'' manner is a fundamental research problem in computational social choice~\cite{FD1}. Classically, this resource allocation problem is referred to as \textsc{Fair Division}, and it has been intensively studied in the literature under various notions of \textit{fairness} and \textit{efficiency}~\cite{aziz2018knowledge,barman2020approximation,bliem2016complexity,bouveret2016characterizing,budish2011combinatorial,caragiannis2019unreasonable,chaudhury2020efx,kurokawa2018fair}. 

It might be desirable, sometimes, that the allocations respect some connectivity measure. For example, while allocating offices to various research teams in a university campus, it might be desirable to provide adjoining offices to the members of the same team. Similarly, while dividing farmlands among heirs, it might be desirable to provide a connected piece to each individual. Another example can be when the items are people and the graph is a social network, and we want each team to consist of members that know each other (not necessarily directly). Thus, Bouveret et al.~\cite{Bouveret} introduced \CFDPO~(\CFDO), an adaptation of \textsc{Fair Division} to graphs. In this problem, we are given (1) a set of items represented as vertices of an input graph $G$, (2) a set of agents $A$, and (3) a set $\mathcal{U}$ of utility functions specifying the utility of each item (vertex) for each agent. The goal is to decide whether there exists an allocation of vertices to the agents which is  (i) \textit{fair} under some notion of fairness, and (ii) \textit{connected} (i.e., for each agent, the subgraph induced by the set of vertices assigned to the agent is connected). Fixing a fairness notion $\varphi$, \CFDO is referred to as $\varphi$-\CFDO. Note that  \textsc{Fair Division} is the special case of \CFDO restricted on complete graphs. Since its introduction, \CFDO~has received a significant amount of interest~\cite{aziz2018knowledge,bilo2022almost,delgikas,igarashi2019pareto}.   

We introduce an ``economic'' generalization of \CFDO, termed \textsc{Incomplete} \CFDO~(\CFD), where we have to allocate exactly $p$ vertices of $G$ to the agents (i.e., $|V(G)|-p$ vertices remain unassigned), and where each agent must be assigned at least one vertex. Note that \CFDO is the special case of \CFD when $p=|V(G)|$. Leaving some items unassigned has recently gained significant attention in context of classical \textsc{Fair Division}~\cite{berger2022almost,knop,caragiannis2019donate,chaudhury2021little}.  Here, the motivation is to attain \textit{fairness} by donating some items to  \textit{charity}. We note that the variation of \CFD~that asks to allocate at most (rather than exactly) $k$ vertices can be modeled using \CFD~(by repeating for $p=1$ to $k$).

\begin{table}
  \centering
  \renewcommand{\arraystretch}{1.2}
  \begin{tabular}{|p{2cm}|c|c|c|c|}
    \hline
    \multirow{2}{5cm}{\textbf{Fairness}} & $p+\mathsf{vcn}+|A|$ & \multicolumn{2}{c|}{$p+\mathsf{vcn}+|A|+\mathsf{val}$}\\
    \cline{2-4}
    & \textbf{Parameterization} & \textbf{Exponential Kernel} & \textbf{No-Poly Kernel}\\
    \hline

    \EF &  \textbf{W[1]-hard} \textbf{[Theorem} \textbf{\ref{th:EFHard}]} &  \textbf{[Theorem} \textbf{\ref{th:Kernel}]} & 
\textbf{[Theorem} \textbf{\ref{th:NoPoly}]} \\ \hline

    \textbf{\EFO} &  \textbf{W[1]-hard} \textbf{[{Theorem}} \textbf{\ref{th:EFOHard}]} &   \textbf{[Theorem} \textbf{\ref{th:Kernel}]} &  \textbf{[Theorem} \textbf{\ref{th:NoPoly}]}\\ \hline
    
    \textbf{\EFX} & \textbf{W[1]-hard} \textbf{[{Theorem}} \textbf{\ref{th:EFOHard}]} &   \textbf{[Theorem} \textbf{\ref{th:Kernel}]}  & \textbf{[Theorem} \textbf{\ref{th:NoPoly}]} \\ \hline

    \textbf{\PROP}  &\textcolor{blue}{\textbf{FPT}}  \textbf{by} $p$ \textbf{[Theorem} \textbf{\ref{th:CCoding}]}  & \textbf{[Theorem} \textbf{\ref{th:Kernel}]} & 
 \textbf{[Theorem} \textbf{\ref{th:NoPoly}]} \\ \hline

  \end{tabular}

  \caption{Summary of our results.}
\end{table}\label{tablepig}

\smallskip
\noindent \textbf{Our Contribution.}
In this paper, we study the parameterized complexity
of \CFD~considering four well-known notions of fairness: \PROP, \EF, \EFO, and \EFX (defined in Section~\ref{S:prelim}). An overview of our results is given in Table~\ref{tablepig}. We remark that parameterized analysis of fair allocations is an extensively studied research subject~\cite{bliem2016complexity,knop,delgikas,FD2,FD3}. 
We first prove that \EF-\CFD remains W[1]-hard when parameterized by $\mathsf{vcn}+p+|A|$, where $\mathsf{vcn}$ is the \textit{vertex cover number} of the item graph and $|A|$ is the total number of agents. To this end, we provide a simple parameterized reduction from the $(k,M)$-\textsc{SUM} problem (parameterized by $k$) to \EF-\CFD~on a star graph parameterized by $p+|A|$. Moreover, we have $|A|=2$, and both agents have \textit{identical valuation} in the reduced instance.

\begin{restatable}{theorem}{EFHard}\label{th:EFHard}
\EF-\CFD is $W[1]$-hard parameterized by $p+|A|$ even for star graphs.
\end{restatable}

Next, we extend this result to the fairness notions \EFO and \EFX. Specifically, we show that \EFO-\CFD and \EFX-\CFD are W[1]-hard when parameterized by $\mathsf{vcn}+p+|A|$, even on graphs having $\mathsf{vcn}=2$. To this end, we provide a non-trivial parameterized reduction from $(k,M)$-\textsc{SUM} problem (parameterized by $k$) to \EFO-\CFD (resp., \EFX-\CFD) parameterized by $p+|A|$ on a graph having $\mathsf{vcn}=2$ (and $|A| = 3$). 

\begin{restatable}{theorem}{EFOHard}\label{th:EFOHard}
\EFO-\CFD and \EFX-\CFD are $W[1]$-hard parameterized by $p+|A|+\mathsf{vcn}$ even for graphs having $\mathsf{vcn}=2$.
\end{restatable}

Since $\varphi$-\CFD, for $\varphi \in \{ \mathsf{EF},\mathsf{EF1},\mathsf{EFX}\}$, is W[1]-hard when parameterized by $p+\mathsf{vcn}+|A|$, we include another parameter in quest of tractability. Let $\mathsf{val}$ be the total number of distinct valuations that agents assign to items ($\mathsf{val}$ is the range of the function $\mathcal{U}$). We note that  $\mathsf{val}$ is used to achieve tractability of \CFDO by Deligkas et al.~\cite{delgikas}. We establish that $\varphi$-\CFD, for $\varphi \in \{ \mathsf{PROP}, \mathsf{EF},\mathsf{EF1},\mathsf{EFX}\}$, is \FPT parameterized by $\mathsf{vcn}+\mathsf{val}+p$  by designing an exponential kernel.

\begin{restatable}{theorem}{Kernel}\label{th:Kernel}
For $\varphi \in \{ \mathsf{EF},\mathsf{EF1},\mathsf{EFX}\}$, $\varphi$-\CFD admits a kernel with at most $p2^{\mathsf{vcn}}\mathsf{val}^{\mathsf{val}}+\mathsf{vcn}$ vertices. Moreover, \PROP-\CFD admits a kernel with at most $p2^{\mathsf{vcn}}\mathsf{val}^{\mathsf{val}}+\mathsf{vcn}+p$ vertices.
\end{restatable}


Next, we complement our exponential kernels by showing that it is unlikely for $\varphi$-\CFD, for $\varphi \in \{ \mathsf{PROP},\mathsf{EF},\mathsf{EF1},\mathsf{EFX}\}$, parameterized by $\mathsf{vcn}+\mathsf{val}+ |A|+p$ to admit a polynomial compression. To this end, we provide \textit{polynomial parameter transformations} from \RBDS.

\begin{restatable}{theorem}{NoPoly}\label{th:NoPoly}
For $\varphi \in \{ \mathsf{PROP}, \mathsf{EF},\mathsf{EF1},\mathsf{EFX}\}$, $\varphi$-\CFD~parameterized by $\mathsf{vcn}+\mathsf{val}+ |A|+p$ does not admit polynomial compression, unless \NPoly.
\end{restatable}

Finally, we establish that \PROP-\CFD is \FPT~when parameterized by $p$ alone. To this end, we provide a color-coding based randomized algorithm with constant success probability\footnote{Clearly, repetition allows to boost the success probability to any constant.}.

\begin{restatable}{theorem}{CCoding}\label{th:CCoding}
There exists a randomized algorithm that solves \PROP-\CFD, in time $e^pp^{\mathcal{O}(p \log p)}m^{\mathcal{O}(1)}$ with success probability at least $1-\frac{1}{e}$. 
\end{restatable}

The choice of $\mathsf{vcn}$ as a structural parameter is quite natural for our results. Arguably, $\mathsf{vcn}$ is the best structural parameter for providing negative results---indeed, our negative results on parameterized complexity and kernelization complexity of \CFD parameterized by $\mathsf{vcn}$ (plus other relevant parameters) imply the same for other smaller parameters such as treewidth, cliquewidth, treedepth, and feedback vertex set (plus other relevant parameters)~\cite{bookParameterized}. Furthermore, $\mathsf{vcn}$ is one of the most efficiently computable parameters from both approximation~\cite{bookApprox} and parameterized~\cite{bookParameterized} points of view, making it fit from an applicative perspective even when a vertex cover is not given along with the input.

\smallskip
\noindent\textbf{Brief Survey.}
The need for fairly dividing resources among competing agents is one of the oldest problems in human civilization. The famous cut-and-choose method to remove envy between two agents dates back to the Book of Genesis. Fairly dividing a divisible item among agents is a very old and classical problem, also referred to as \textit{cake-cutting}, and is studied extensively in the literature~\cite{brams1996fair,sonmez2010course}.

When the items are indivisible, the problem is well studied under the name \textsc{Fair Division}~\cite{aziz2018knowledge,barman2020approximation,bliem2016complexity,bouveret2016characterizing,brams2003fair,budish2011combinatorial,caragiannis2019unreasonable,chaudhury2020efx,kurokawa2018fair}. 
The study of \CFDO~was initiated by Bouveret et al.~\cite{Bouveret}. Later, their work was extended to include the notion of \textit{chores} as well by Aziz et al.~\cite{aziz2018knowledge}. Various notions of fairness like maximin share allocation (Greco and Scarcello~\cite{greco2020complexity}) and Pareto-optimal allocations (Igarashi and Peters~\cite{igarashi2019pareto}) are also studied. Bil{\`o} et al.~\cite{bilo2022almost} studied the conditions that guarantee various notions of fair allocations. Recently, Deligkas et al.~\cite{delgikas} provided a comprehensive picture of the parameterized complexity of \CFDO.

\smallskip
\noindent\textbf{Motivation.} A wide range of considerations motivates our definition of \CFD. First, it is often the case that the allocator wants to save some resources for later. For instance, consider the example of a university campus where the professors and research groups are first allocated offices while saving some rooms  for various administrative purposes and classrooms, or for future professors in case the university is expanding. In these settings, the requirement of connectivity might be desirable.    

Second, partial allocations are specifically useful when the allocator wants to save some resources for the future. For instance, it is  practical for a university department to not allocate all of its resources (like computer equipment and travel funding) in one round of allocations and save some for future use. The notion of connectivity can be introduced based on the graph of devices that work well with each other. A similar setting arises while allocating ground to companies/people to build buildings while saving some ground for the future. Here, the connectivity requirement comes naturally.

Third, it saves resources from the viewpoint of the allocators: using this framework, they can assign the least amount of resources that makes the agents ``happy''. This is best exemplified when the allocator actually needs to buy the items. For example, when a company actually needs to rent offices in an office space or has to buy some equipment for its employees.

For another example, consider a park where different groups want to organize picnics on a specific day, and the park committee has to allocate these picnic spots. Allocating a connected spot to each group makes good sense here. Further, each day, the committee receives multiple applications with the preferences of each group. Now, it may be impractical to delay the allocation process till the very last day. So, it may be a good idea to allocate these spots in phases such that in each phase, the applicants are provided the spots ``fairly''  while saving a sufficient amount of spots for further phases.

Due to the connectivity requirement, it might so happen that there are some ``problematic'' vertices, which, when assigned to one of the agents, may deem the allocation ``unfair'' in \CFDO. For example, consider a vertex $v$ such that $\ell$ degree-one vertices are attached to it. Moreover, let $n$ be the number of agents, and $\ell$ is much larger than $n$. In this case, since each vertex is to be assigned to some agent in \CFDO, at least $\ell-n$ of these degree-one vertices must be assigned to the same agent that is assigned the vertex $v$, possibly making other agents ``envious'' of this agent.  In these scenarios, it is a practical (and desirable) question to seek a fair allocation by leaving some vertices unassigned.

\section{Preliminaries}\label{S:prelim}
For $\ell \in \mathbb{N}$, let $[\ell]=\{1,\ldots, \ell\}$ and  $[0,\ell]=\set{0}\cup [\ell]$. Let $\mathbb{N}_0=\mathbb{N}\cup \set{0}$.

\smallskip
\noindent\textbf{Graph Theory.} We consider finite, simple, and connected graphs. For a graph $G$, let $V(G)$ and $E(G)$ denote the \textit{vertex set} and the \textit{edge set} of $G$, respectively. For a vertex $v\in V(G)$, let $N(v)$ denote the \textit{open neighborhood} of $v$, that is, $N(v) = \{u~|~uv\in E(G)\}$. For a subset $X\subseteq V(G)$, let $N_X(v)=N(v)\cap X$. For a subset $X\subseteq V(G)$, let $G[X]$ denote the subgraph of $G$ induced by vertices in $X$. Moreover, let $G-X$ denote the subgraph $G[V(G)\setminus X]$ of $G$. A set $U \subseteq V(G)$ is a \textit{vertex cover} of $G$ if for every edge in $G$, at least one of its endpoints is in $U$. The minimum cardinality of a vertex cover of $G$ is its \textit{vertex cover number} ($\mathsf{vcn}$). Given a graph $G$, let $d_G(v)$ be the \textit{degree} of the vertex $v$ in $G$, i.e., $d_G(v) = |N(v)|$. For standard graph theoretic terminology not defined explicitly in this paper, we refer to the book by Diestel~\cite{diestel}.

\smallskip
\noindent\textbf{Incomplete Connected Fair Division.} An instance of \CFDP (\CFD) consists of $(G,A,$ $\mathcal{U},p)$ where $G$ is the \textit{utility graph} on $m$ vertices, $A = [n]$ is the set of $n$ \textit{agents}, $\mathcal{U}$ is the set of \textit{utility functions} $\{u_i: V(G) \rightarrow \mathbb{N}_0\}_{i\in [n]}$, and $p \in  [m]$. When clear from context, we denote $|V(G)|$ by $m$ and $|A|$ by $n$. Every vertex $v \in V(G)$ corresponds to an \textit{item}. It is a standard assumption in the literature, and we assume it too, that the agents have \textit{additive valuations}, i.e., for every $X\subseteq V(G)$ and for every $i\in A$, we have $u_i(X) = \sum_{v\in X}u_i(v)$. An \textit{allocation} of items is a tuple $\Pi = (\pi_1,\ldots, \pi_n)$ such that:
\begin{enumerate}
    \item For $i\in A$: $\pi_i \subseteq V(G)$, $|\pi_i| \geq 1$, and $G[\pi_i]$ is connected.
    \item For $i,j\in A$ such that $i\neq j$: $\pi_i\cap\pi_j = \emptyset$.
    \item $|\bigcup_{i\in A}\pi_i| = p$.
\end{enumerate}

We say that the \textit{bundle} $\pi_i$ is \textit{assigned} to agent $i$ in $\Pi$. For a bundle $\pi_i$, let $\tau_i = \{v\in \pi_i~|~ G[\pi_i\setminus \{v\}]\text{ is connected}\}$. Next, we have the following notions of \textit{fairness} (that we consider in this paper). An allocation $\Pi = (\pi_1,\ldots, \pi_n)$ is:

\begin{itemize}
    \item \textit{proportional} (\PROP) if for every $i\in A$, $u_i(\pi_i) \geq \frac{u_i(V(G))}{n}$;
    \item \textit{envy-free} (\EF) if for every $i,j \in A$, $u_i(\pi_i) \geq u_i(\pi_j)$;
    \item \textit{envy-free up to one item} (\EFO) if for every $i,j\in A$, $u_i(\pi_i) \geq u_i(\pi_j) - \max_{v\in \tau_j}u_i(v)$;
    \item \textit{envy-free up to any item} (\EFX) if for every $i,j\in A$, $u_i(\pi_i) \geq u_i(\pi_j) - \min_{v\in \tau_j}u_i(v)$.
\end{itemize}

For a fairness criterion $\varphi \in \{\mathsf{PROP}, \mathsf{EF}, \mathsf{EF1}, \mathsf{EFX}\}$, $\varphi-$\CFD~asks whether there exists an allocation $\Pi$ for $(G,A,\mathcal{U},p)$ satisfying $\varphi$. An allocation $\Pi$ that satisfies $\varphi$ is termed as $\varphi$-\textit{allocation}.
Agents $i$ and $j$ have \textit{same type} if for every $v\in V(G)$, $u_i(v) = u_j(v)$. Let $\mathcal{A}$ denote the set of all agent types, i.e., $\mathcal{A}$ is the partition of  $A$ such that agents $i$ and $j$ are in the same part $\mathsf{a} \in \mathcal{A}$ if and only if $i$ and $j$ have the same type. The type of an agent $i$ is the type $\mathsf{a}\in \mathcal{A}$ such that $i\in \mathsf{a}$; let $\mathsf{a}_i$ denote the type of agent $i$. Since the agents of each type have the same valuation function, we let $u_{\mathsf{a}}$ denote the valuation function $u_i$ where $i\in \mathsf{a}$. When $|\mathcal{A}| = 1$, we say that the agents have \textit{identical valuations} and for each vertex $v$ and agent $i$, we also denote $u_i(v)$ by $u(v)$ for brevity.
We say that an instance $(G,A,\mathcal{U},p)$ admits \textit{unary valuation} (resp., \textit{binary valuation}) if the valuations (in the valuation function $\mathcal{U}$) are encoded in unary (resp., binary), i.e., there is some polynomial (resp., exponential) function $f$ such that $\max_{i\in A, v\in V(G)}$ $u_i(v) \leq f(|V(G)|)$.

\smallskip
\noindent\textbf{Parameterized Complexity.}
In the framework of parameterized complexity, each instance of a problem $\Pi$ is associated with a non-negative integer \textit{parameter} $k$. A parametrized problem $\Pi_p$ is \textit{fixed-parameter tractable} ($\mathsf{FPT}$) if there is an algorithm that, given an instance $(I,k)$ of $\Pi_p$, solves it in time $f(k)\cdot |I|^{\mathcal{O}(1)}$, for some computable function $f(\cdot)$. Central to parameterized
complexity is the following hierarchy of complexity classes:
$
\mathsf{FPT} \subseteq \mathsf{W[1]}  \subseteq \mathsf{W[2]} \subseteq \ldots \subseteq \mathsf{XP}.
$
We note that \FPT $\neq$ W[1] under $\mathsf{ETH}$.

Two instances $I$ and $I'$ (possibly of different problems) are  \textit{equivalent} when $I$ is a Yes-instance if and only if $I'$ is a Yes-instance. A \textit{compression} of a parameterized problem $\Pi_1$ into a (possibly non-parameterized) problem $\Pi_2$ is a polynomial-time algorithm that maps each instance $(I,k)$ of $\Pi_1$ to an equivalent instance $I'$ of $\Pi_2$ such that size of $I'$ is bounded by $g(k)$, for a computable function $g(\cdot)$. If $g(\cdot)$ is polynomial, then the problem is said to admit a \textit{polynomial compression}. When $\Pi_1 = \Pi_2$, compression is also termed as {\em kernelization}.

A \textit{polynomial parameter transformation} from $\Pi_1$ to $\Pi_2$ is a polynomial time algorithm that given an instance $(I,k)$ of $\Pi_1$ generates an equivalent instance $(I',k')$ of $\Pi_2$
such that $k' \leq p(k)$, for some polynomial $p(\cdot)$. Here, if $\Pi_1$ does not admit a polynomial compression, then $\Pi_2$ does not admit a polynomial compression~\cite{bookParameterized}. We refer to the books by Cygan et al.~\cite{bookParameterized} and Fomin et al.~\cite{bookKernelization} for more details on parameterized complexity.

\subsection{Preliminary Results and Observations}\label{SS:P}
Consider an instance $(G,A,\mathcal{U},p)$ of \CFD. For each $i\in A$, we assume without loss of generality that there is at least one vertex $v$ such that $u_i(v) > 0$. Moreover, since each agent must be assigned at least one vertex, we can safely assume that $|A| >p$, otherwise $\varphi$-\CFD~becomes trivial for $\varphi \in \{\mathsf{PROP}, \mathsf{EF},\mathsf{EF1},\mathsf{EFX}\}$. We have the following basic observations and lemmas.


\begin{observation}\label{O:identical}
Let $\Pi$ be an allocation for an instance $(G,A,$ $\mathcal{U},p)$ of \EF-\CFD~with identical valuations (i.e., $|\mathcal{A}| = 1$). Then, $\Pi$ admits \EF if and only if for any two agents $i,j\in A$, $u(\pi_i)=u(\pi_j)$. 
\end{observation}
\begin{proof}
In one direction, suppose $\Pi$ satisfies \EF. Targeting contradiction, assume that there are two agents $i$ and $j$ such that $u(\pi_i)\neq u(\pi_j)$. Without loss of generality assume that $u(\pi_i)< u(\pi_j)$. Then, observe that the agent $i$ is envious of the agent $j$.

The reverse direction is rather easy to see as it directly satisfies the condition for \EF~(i.e., for every $i,j\in A$, $u_i(\pi_i) \geq u_i(\pi_j)$). 
  \end{proof}


\begin{lemma}\label{L:subgraph}
Consider an instance $(G,A,\mathcal{U},p)$ of $\varphi$-\CFD~where $\varphi \in \{ \mathsf{EF},\mathsf{EF1},\mathsf{EFX}\}$. Let $H$ be an induced subgraph of $G$. If $(H,A,\mathcal{U},p)$ is a Yes-instance of $\varphi$-\CFD, then $(G,A,\mathcal{U},p)$ is a Yes-instance of $\varphi$-\CFD. 
\end{lemma}
\begin{proof}
Let $\Pi = (\pi_1, \ldots, \pi_n)$ be an allocation for  $(H,A,\mathcal{U},p)$ that satisfies $\varphi$ ($\varphi \in \{ \mathsf{EF},\mathsf{EF1},\mathsf{EFX}\}$). First, observe that $\Pi$ is also an allocation for $(G,A,\mathcal{U},p)$ (since for $i\in A$, each $G[\pi_i]$ is connected, $|\cup \pi_i|=p$, and for distinct $i,j\in A$, $\pi_i \cap \pi_j =\emptyset)$. Moreover, it is easy to see that the allocation $\Pi$ for $(G,A,\mathcal{U},p)$ satisfies $\varphi$.
  \end{proof}

We remark here that for an induced subgraph $H$ of $G$, it might happen that $(H,A,\mathcal{U},p)$ is a Yes-instance of \PROP-\CFD~but $(G,A,\mathcal{U},p)$ is a No-instance of \PROP-\CFD. To see this, consider a Yes-instance $(H,A,\mathcal{U},p)$ of \PROP-\CFD~such that $H = G-\{v\}$ and $|A|>1$. Now, for each agent $i\in A$, let $u_i(v) = (n+1)\sum_{x\in V(H)} u_i(x)$. Now, consider the instance $(G,A,\mathcal{U},p)$. Since $|A|>1$, there is at least one agent, say, $j$, such that the vertex $v$ is not assigned to the agent $j$. Observe that even if $\pi_j = V(G)\setminus \set{v}$, $u_j(\pi_j) < \frac{u_j(V(G))}{n}$ (assuming $u_j(V(H)) > 0$). Hence, $(G,A,\mathcal{U},p)$ is a No-instance. 
Next, we have the following observation concerning \PROP-\CFD~that we will use later.

\begin{observation}\label{O:proper}
Let $p<p' \leq |V(G)|$. Moreover, let $(G,A,\mathcal{U},p)$ be a Yes-instance of \PROP-\CFD. Then, $(G,A,\mathcal{U},p')$ is a Yes-instance of \PROP-\CFD~as well.
\end{observation}
\begin{proof}
We will prove this using an inductive argument. The base case $(G,A,\mathcal{U},p)$ is a Yes-instance by our assumption. Next, let $(G,A,\mathcal{U},k)$ is a Yes-instance for some $p\leq k \leq m$ (where $m = |V(G)|$).  Then, we show that  $(G,A,\mathcal{U},k+1)$ is a Yes-instance (where $k+1\leq m$). Let $\Pi=(\pi_1,\ldots,\pi_n)$ be a \PROP~allocation for $(G,A,\mathcal{U},k)$. Since $G$ is a connected graph and $k<m$, observe that there is at least one vertex $v$ such that $v$ is not assigned to any agent in $\Pi$, and there is some $\pi_i\in \Pi$ such that $G[\pi_i\cup v]$ is connected. Observe that $\Pi'=(\pi'_1,\ldots,\pi'_n)$ where $\pi'_i=\pi'_i\cup \set{v}$ and for $j\neq i$, $\pi'_j=\pi_j$ is an allocation for $(G,A,\mathcal{U},k+1)$ satisfying \PROP. This completes our proof.
  \end{proof}


Finally, consider the following trivial observation that follows directly from the definition of \EFO-\CFD and \EFX-\CFD.

\begin{observation}\label{O:PreEFO}
Consider an allocation $\Pi = (\pi_1,\ldots,\pi_n)$ for an instance $(G,A,\mathcal{U},p)$. If there are two agents $i$ and $j$ such that for each vertex $v\in \tau_j$, $u_i(\pi_i) < u_i(\pi_j)- u_i(v)$, then  $\Pi$ is neither an \EFO allocation nor an \EFX allocation. 
\end{observation}

\section{W[1]-hardness Results}
In this section, we show that $\varphi$-\CFD~where $\varphi \in \{ \mathsf{EF},\mathsf{EF1},\mathsf{EFX}\}$ is W[1]-hard when parameterized by $p+\mathsf{vcn}+|A|$. 
For this purpose, we first define the following problem. The \ksum problem is to determine, given $N$ integers $a_1,\ldots, a_N \in [0,M]$ and a target integer $t\in [0,M]$, whether there exists $S\subseteq N$ such that $|S|=k$ and $\sum_{i\in S} a_i = t$. We also assume that $2\leq k\leq N-2$ (as otherwise, we can solve the problem in polynomial time). Abboud et al.~\cite{kSUM} proved that \ksum is W[1]-complete parameterized by $k$ even when $M= N^{2k}$:

\begin{proposition}[\cite{kSUM}]\label{P:ksum}
$(k,N^{2k})$-\textsc{SUM} parameterized by $k$ is $W[1]$-complete.
\end{proposition}

\subsection{\EF-\CFD}
First, we prove that \EF-\CFD is W[1]-hard parameterized by $p+\mathsf{vcn}+|A|$. To this end, we provide a parameterized reduction from \ksum to \EF-\CFD~on a star graph such that $p = k+2$ and $|A|=2$. First, we explain our construction.

\medskip
\noindent\textbf{Construction.} Let $(a_1,\ldots,a_N,t)$ be an instance of \ksum. We create an instance $(G,A,\mathcal{U},p)$ of \EF-\CFD in the following manner. Let $G$ be a star on $N+2$ vertices such that $V(G) = \{v_1,\ldots, v_N\} \cup \{d_1,d_2\}$ and $E(G) = \{d_1d_2\} \cup \set{d_1v_i~|~i\in[N]}$ (i.e., $G$ is a star with $N+1$ leaves and $d_1$ as the center vertex). Let $A = \set{1,2}$. Moreover, for $j\in A$ and $i\in [N]$, let $u_j(v_i) = a_i$, $u_j(d_1)=1+ \sum a_i$, and $u_j(d_2) = u_j(d_1)+t$ (here, $|\mathcal{A}|=1$). Finally, we set $p=k+2$. Since $k\geq 2$, note that $p\geq 4$. 

Now, we have the following lemma.
\begin{lemma}\label{L:ksum}
$(a_1,\ldots,a_N,t)$ is a Yes-instance of \ksum if and only if $(G,A,\mathcal{U},p)$ is a Yes-instance of \EF-\CFD.
\end{lemma}
\begin{proof}
In one direction, suppose $(a_1,\ldots,a_N,t)$ is a Yes-instance of \ksum and $S\subseteq [N]$ is a set such that $|S|=k$ and $\sum_{i\in S} a_i = t$. Consider the allocation $\Pi = (\pi_1,\pi_2)$ as $\pi_1 = \set{d_1} \cup \set{v_i~|~ i\in S}$ and $\pi_2 = d_2$. 
Note that here $u(\pi_1) = u(\pi_2) = t+\sum_{i\in [N]} a_i$. Since the valuations are identical, due to Observation~\ref{O:identical}, $\Pi$ satisfies \EF. 

In the other direction, suppose $(G,A,\mathcal{U},p)$ is a Yes-instance of \EF-\CFD~and $\Pi=(\pi_1,\pi_2)$ is an allocation for $(G,A,\mathcal{U},p)$ satisfying \EF. Since $G$ is a star graph and $p\geq 4$, note that at least one of the agents is assigned at least two vertices and that agent is assigned vertex $d_1$ as well. Without loss of generality, let us assume that $d_1\in \pi_1$. Then, note that $u(\pi_1) \geq 1+\sum_{i\in [N]} u(v_i)$. Now, we claim that the vertex $d_2\in \pi_2$. Targeting contradiction, assume that $d_2\notin \pi_2$. Then, $u(\pi_2) \leq \sum_{i\in [N]} u(v_i) < u(\pi_1)$. This contradicts the fact that $\Pi$ satisfies \EF~(due to Observation~\ref{O:identical}). 

Hence, $d_2\in \pi_2$. Since $d_1$ is assigned to $\pi_1$, note that no other vertex can be assigned to $\pi_2$ while keeping $G[\pi_2]$ connected. Hence, $\pi_2 = \set{d_2}$ and $u(\pi_2) = u(d_1)+t$. Since $|\pi_2| = 1$, note that $|\pi_1|= p-1 = k+1$. Hence, exactly $k$ vertices are assigned to $\pi_i$ other than $d_1$. Moreover, since $\Pi$ satisfies \EF, due to Observation~\ref{O:identical}, we have $u(\pi_1) = u(\pi_2)$. Therefore, there is a set $S\subseteq [N]$ such that $|S| = k$ and $\sum_{i\in S} u(v_i) = u(\pi_2)-u(d_1) = t$. This implies that $\sum_{i\in S} a_i = t$ (since $a_i = u(v_i)$). Hence $(a_1,\ldots, a_N,t)$ is a Yes-instance. 
  \end{proof}

We have the following theorem as a consequence of  Proposition~\ref{P:ksum} and Lemma~\ref{L:ksum}.

\EFHard*
\begin{proof}
Note that in our construction, $G$ is a star graph, $p=k+2$, and $|A|=2$. The rest of the proof follows from our construction, Proposition~\ref{P:ksum}, and Lemma~\ref{L:ksum}. Moreover, the problem remains W[1]-hard parameterized by $p+|A|$ even when the maximum valuation is bounded by $n^{2p}$ (binary valuations).
  \end{proof}

\subsection{\EFO-\CFD and \EFX-\CFD}
Now, we prove that \EFO-\CFD and \EFX-\CFD are W[1]-hard when parameterized by $p+\mathsf{vcn}+|A|$, even when both $\mathsf{vcn}$ and $|A|$ are small constants. Here also, we provide a parameterized reduction from \ksum to \EFO-\CFD (resp.,\EFX-\CFD) on a graph with $\mathsf{vcn}= 2$, while $|A| = 3$, and $p = k+6$. First, we explain our construction. 

\begin{figure}
    \centering
    \includegraphics{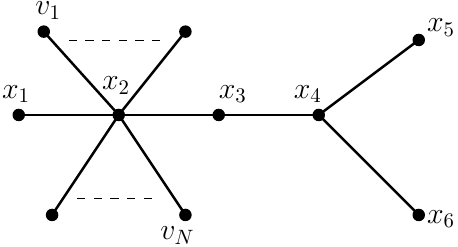}
    \caption{Illustration of the construction of graph $G$.}
    \label{fig:construction}
\end{figure}
\medskip
\noindent\textbf{Construction.} Let $(a_1, \ldots, a_N,t)$ be an instance of \ksum. First, we will define how to construct the graph $G$ as it is the same for both \EFO-\CFD and \EFX-\CFD.  Consider a path on five vertices $x_1,\ldots,x_5$ and attach a vertex $x_6$ to $x_4$. Now, we add $N$ vertices $v_1,\ldots,v_N$ and attach them to $x_2$ by an edge. More formally, let $V(G) = \{ x_1,\ldots, x_6, v_1,\ldots, v_N\}$ and $E(G) = \{x_2v_i~|~i\in [N]\} \cup \{x_4x_6 \} \cup \{x_ix_{i+1}~|~1\leq i\leq 4 \}$. See Figure~\ref{fig:construction} for a reference. Next, we have $A=[3]$ and $p=k+6$. Finally, we define the valuation functions $\mathcal{U}$ and $\mathcal{U}'$ for \EFO-\CFD and \EFX-\CFD, respectively. These functions will be identical except for the valuation $u_3(x_1)$ and $u'_3(x_1)$. Moreover, let $C= \sum_{i\in [N]} a_i$. The valuation functions are defined as follows:
\begin{itemize}
    \item $u_i(v_j) = C+ a_i$ for $i\in [3]$ and $j\in [N]$,
    \item $u_i(x_2) = 2NC$ for $i\in [3]$,
    \item $u_1(x_3)=u_2(x_3) = 0$ and $u_3(x_3)=2NC+kC+t$,
    \item $u_1(x_4)=u_2(x_4) = 3NC+kC+t$ and $u_3(x_4)=0$,
    \item $u_i(x_5) = 0$ for $i\in [3]$,
    \item $u_i(x_6) = 0$ for $i\in [3]$,
    \item $u_i(x_1) = NC$ for $i\in [3]$. Finally, $u'_3(x_1)= 0$.
\end{itemize}

First, we have the following straightforward lemma that proves one side of our reduction.
\begin{lemma}\label{LF:EFX}
    If $(a_1,\ldots, a_N,t)$ is a Yes-instance of \ksum, then $(G,A,\mathcal{U}, p)$ and $(G,A,\mathcal{U}',p)$ are Yes-instances of \EFO-\CFD and \EFX-\CFD, respectively.
\end{lemma}
\begin{proof}
    Let $S\subseteq \{a_1,\ldots a_N\}$ be a set such that $\sum_{a\in S}a = t$ and $|S|=k$. Then, consider the following allocation $\Pi = (\pi_1,\pi_2,\pi_3)$:
    \begin{itemize}
        \item $\pi_1 = S\cup \{x_1,x_2\}$
        \item $\pi_2 = \{x_4,x_5,x_6\}$
        \item $\pi_3 = \{x_3\}$
    \end{itemize}

    It is easy to see that $|\pi_1|+|\pi_2|+|\pi_3| = |S|+6=k+6 = p$. Furthermore, it is not difficult to verify that $\Pi$ is an \EFO-\CFD and \EFX-\CFD for $(G,A,\mathcal{U}, p)$ and $(G,A,\mathcal{U}',p)$, respectively. For completeness, we give explicit details below.
    \begin{enumerate}
        \item $\Pi$ is an \EFO-\CFD for $(G,A,\mathcal{U}, p)$: Here, observe that $u_1(\pi_1) = u_1(x_1)+u_1(x_2) + u_1(S) = NC+2NC+kC+t= 3NC+kC+t$. Since $u_1(\pi_2) = u_1(x_4)+u_1(x_5)+u_1(x_6)=3NC+kC+t$ and $u_1(\pi_3)=u_1(x_3)=0$, Agent~1 is not envious of any other agent. Similarly, $u_2(\pi_2) = 3NC+kC+t \geq u_2(\pi_1) =3NC+kC+t > 0=u_2(\pi_3)$. Hence, Agent~2 is not envious of any other agent. Similarly, $u_3(\pi_3) = 2NC+kC+t > 0= u_3(\pi_2)$. Hence, Agent~3 is not envious of Agent~2. Now, the only thing remaining to prove is that Agent~3 is not envious of Agent~1 up to one item, i.e., $u_3(\pi_3) \geq u_3(\pi_1)- \max_{v\in \tau_1}u_3(v)$. Now, observe that $x_1\in \tau_1$ since $x_1$ is a degree-one vertex. Hence, $u_3(\pi_1)- \max_{v\in \tau_1}u_3(v) \leq u_3(\pi_1) - u_3(x_1)= 2NC+kC+t$.  Therefore, $u_3(\pi_2) \geq u_3(\pi_1)- \max_{v\in \tau_1}u_3(v)$, which completes our argument.

        \item $\Pi$ is an \EFX-\CFD for $(G,A,\mathcal{U}',p)$: This proof is similar to the proof of the previous case. Since $\mathcal{U}$ and $\mathcal{U}'$ differ only in the valuation of $u_3(x_1)$ and $u'_3(x_1)$, it follows using the arguments in the proof of the previous case that Agents~1 and 2 are not envious of Agent~3 and Agent~3 is not envious of Agent~2. Moreover, since $u'_3(\pi_3) = 2NC+kC+t$ and $u'_3(\pi_1) = u'_3(x_1)+u'_3(x_2)+u'_3(S) = 0+2NC+kC+t$, we have that Agent~3 is not envious of Agent~2 as well. 
    \end{enumerate}
    This completes our proof.   
\end{proof}

\subsubsection{Some Useful Observations.}
Here, we prove some observations that will be useful for proving the other direction of  our reduction.  For the observations where we do not use any of $u_3(x_1)$ and $u'_3(x_1)$, we will use only the valuation functions $\mathcal{U}$ to ease the presentation. Moreover, we remark that these observations are valid for both \EFO-\CFD (assuming valuations $\mathcal{U}$) and \EFX-\CFD (assuming valuations $\mathcal{U}'$). First, observe that agents 1 and 2 have identical valuations. Moreover, since $k\geq 2$, note that $p\geq 8$. Furthermore, observe that $\{x_2,x_4\}$ is a vertex cover of $G$ and hence $G[V(G)\setminus \{x_2,x_4\}]$ is an independent set.

For the rest of this section, let $\Pi= (\pi_1,\pi_2,\pi_3)$ be an allocation that is an \EFO-\CFD for $(G,A,\mathcal{U},p)$ (resp., \EFX-\CFD for $(G,A,\mathcal{U}',p)$).
First, we have the following easy observation.
\begin{observation}\label{O:notInTau}
    If $x_2\in \pi_i$, for $i\in [3]$, then $x_2 \notin \tau_i$.
\end{observation}
\begin{proof}
Since $G$ is a tree, note that $G[\pi_i]$ is also a tree. Observe that removal of any vertex that has degree at least two from a tree gives at least two connected components. Hence, to prove our claim, it suffices to show that $d_{G[\pi_i]}(x_2)>1$. 

Recall that $|\pi_1|+|\pi_2|+|\pi_3| \geq 8$. We will not use the valuation functions for this proof, and hence, without loss of generality, we can assume that $x_2\in \pi_1$.  Let $\alpha = |\pi_1 \cap \{x_3,x_4,x_5,x_6\}|$ ($\alpha$ can be equal to $0$). Then, due to connectivity constraints, observe that $|\pi_2|+|\pi_3| \leq 4- \alpha +1$ as agent 2 (resp., agent 3) can be either assigned a vertex from $\{v_1,\ldots,v_N\}\cup \{x_1\}$ or at most $4-\alpha$ vertices from $\{x_3,\ldots,x_6\}$. Now, $|\pi_1|+|\pi_2|+|\pi_3| \geq 8$, and hence, $|\pi_1| \geq 8- (|\pi_2|+|\pi_3|) \geq 8 - (5- \alpha)$. Thus, $|\pi_1|\geq 3+\alpha$. Therefore, irrespective of the value of $\alpha$, $\pi_1$ contains at least two vertices from $\{x_1,v_1,\ldots, v_N\}$, and hence $d_{G[\pi_1]}(x_2)>1$. This completes our proof. 
\end{proof}

Next, we have the following observation.

\begin{observation}\label{O:Hard1}
    $x_2\in \pi_1\cup \pi_2$.
\end{observation}
\begin{proof}
    First, we will prove that $x_2\notin \pi_3$. Targeting contradiction assume that $x_2\in \pi_3$.  Thus, due to Observation~\ref{O:notInTau}, $x_2\notin \tau_3$. Hence, we have that $u_1(\pi_3)- \max_{v\in \tau_3}(u_1(v)) \geq u_1(x_2)$. Therefore, to satisfy \EFX (resp., \EFO) $u_1(\pi_1) \geq u_1(x_2) = 2NC$. By an identical argument, $u_2(\pi_2) \geq 2NC$. Now, notice that the only way such that $u_1(\pi_1) \geq 2NC$ is if $x_4\in \pi_1$. Now, since $G[V(G)\setminus \{x_2,x_4\}]$ is an independent set, $\pi_2$ consists of exactly one vertex, say $y$, such that $y\notin \{x_2,x_4\}$. Since for each such $y$, $u_2(y)<2NC$, this contradicts the fact that $\Pi$ respects \EFX (resp., \EFO). 

    Therefore, we have that $x_2\notin \pi_3$. Finally, to see that $x_2\in \pi_1\cup \pi_2$, observe that $G[V(G)\setminus \{x_2\}]$ contains one connected component of size\footnote{\textit{Size} of a connected component is the number of vertices in it.} at most $4$, and each remaining  component has size at most $1$.  Hence, if $x_2\notin \pi_1\cup \pi_2 \cup \pi_3$, then $|\pi_1|+|\pi_2|+|\pi_3| \leq 4+1+1 = 6$, a contradiction to the fact that $p\geq 8$. Therefore, $x_2 \in \pi_1\cup \pi_2$.       
\end{proof}

Since agents 1 and 2 have identical valuations and $x_2\in \pi_1\cup \pi_2$ (due to Observation~\ref{O:Hard1}), for the rest of this section, we can assume without loss of generality that $x_2\in \pi_1$. We have the following remark.

\begin{remark}\label{remark}
    For the rest of this section, we assume that $x_2\in \pi_1$.
\end{remark}

Next, we have the following observation.
\begin{observation}\label{O:Hard2}
    $\pi_3 = \{x_3\}$ and $\pi_2 = \{x_4,x_5,x_6\}$.
\end{observation}
\begin{proof}
    First, we have the following claim.
    
    \begin{clm}\label{C:4in2}
        $x_4\in \pi_2$.
    \end{clm}
    \begin{proofofclaim}
    Targeting contradiction, assume that $x_4 \notin \pi_2$. Since $G[V(G)\setminus \{x_2,x_4\}]$ is an independent set and $x_2,x_4 \notin \pi_2$ (as, due to Remark~\ref{remark}, $x_2\in \pi_1$), $\pi_2$ consists of exactly one vertex, say $y$, such that $y \notin \{x_2,x_4\}$. Moreover, due to Observation~\ref{O:notInTau}, $x_2\notin \tau_1$. Therefore, if $\Pi$ respects either of \EFO or \EFX, we should have that $u_2(y) \geq u_2(x_2)$. But, since no such $y$ exists, we reach a contradiction. Hence, $x_4 \in \pi_2$. 
    \end{proofofclaim}

    Second, we have the following claim.
    
    \begin{clm}\label{C:3in3}
    $\pi_3 = \{ x_3\}$.
    \end{clm}
    \begin{proofofclaim}
        Since $x_2\in \pi_1$ (due to Remark~\ref{remark}), $x_4\in \pi_2$ (due to Claim~\ref{C:4in2}), and $G[V(G)\setminus \{ x_2,x_4\}]$ is an independent set, $\pi_3$ consists of a single vertex. Moreover, since $x_2\notin \tau_1$ (due to Observation~\ref{O:notInTau}), to respect either of \EFO or \EFX, $u_3(\pi_3) \geq u_3(x_2) = 2NC$. This is only possible if $\pi_3 = \{x_3\}$.\end{proofofclaim}

    Finally, we show that $\pi_2 = \{x_4,x_5,x_6\}$. Observe that if we can prove that $|\pi_1| \leq k+2$, then, since $|\pi_3| = 1$, we will have that $|\pi_2| \geq p-k-3  = 3$. This will be only possible if $|\pi_2| = 3$ and $\pi_2 = \{x_4,x_5,x_6\}$. Hence, we prove the following claim, which in turn establishes that $\pi_2 = \{x_4,x_5,x_6\}$ 

    \begin{clm}\label{C:size1}
        $|\pi_1| \leq k+2$. 
    \end{clm}
    \begin{proofofclaim}
    Targeting contradiction assume that $|\pi_1| \geq k+3$. Recall that $|\pi_3| = \{x_3\}$, and hence, $u_3(\pi_3) =  2NC+kC+t$. Moreover, recall that $x_2 \in \pi_1$ (due to Remark~\ref{remark}) and $x_2\notin \tau_1$ (due to Observation~\ref{O:notInTau}). Moreover, observe that $G[\pi_1]$ induces a star graph, and hence, $\tau_1 = \pi_1\setminus \{x_2\}$. Therefore, the following equations follow directly: (i) $\pi_1 = \tau_1 \cup \{x_2\}$ and (ii) $|\tau_1| \geq k+2$. Moreover, $\tau_1 \subseteq \{x_1,v_1,\ldots, v_N\}$. 

    First, observe that for any $S\subseteq \{ v_1,\ldots, v_N\}$, $u_3(S)+u_3(x_2) > 2NC+(k+1)C > 2NC+kC+t = u_3(\pi_3)$ (since $t<C$). Hence, if $\tau_1 \subseteq \{v_1,\ldots, v_N\}$ (i.e., $x_1\notin \tau_1$), then, since $|\tau_1|>k+1$, for any vertex $v\in \tau_1$, $u_3(\pi_1)-u_3(v) > u_3(\pi_3)$, which contradicts that $\Pi$ is an \EFO-\CFD (resp., \EFX-\CFD) for $(G,A,\mathcal{U},p)$ (resp., $(G,A,\mathcal{U}',p)$). Finally, for the case when $x_1 \in \tau_1$, it is sufficient to observe that for any vertex $v \in \{v_1,\ldots, v_N\}$, $u_3(x_1) > u_3(v)$ and $u'_3(x_1) < u'_3(v)$. Thus, we have  $u_3(\pi_1)- \max_{v\in \tau_1} u_3(v) > u_3(\pi_3)$ and $u'_3(\pi_1)- \min_{v\in \tau_1} u'_3(v) > u'_3(\pi_3)$. This contradicts the fact that $\Pi$ is an \EFO-\CFD (resp., \EFX-\CFD) for $(G,A,\mathcal{U},p)$ (resp., $(G,A,\mathcal{U}',p)$). 
    \end{proofofclaim}
This completes our proof. 
\end{proof}

Finally, we are ready to prove our main lemma, which implies the other side of the reduction.
\begin{lemma}\label{L:EFX2}
If $(G,A,\mathcal{U},p)$ (resp., $(G,A,\mathcal{U}',p)$) is a Yes-instance of \EFO-\CFD (resp., \EFX-\CFD), then $(a_1,\ldots,a_N,t)$ is a Yes-instance of \ksum.
\end{lemma}
\begin{proof}
    Let $\Pi$ be an allocation that satisfies \EFO (resp., \EFX) for $(G,A,\mathcal{U},p)$ (resp., $(G,A,\mathcal{U}',p)$). Then, due to Observation~\ref{O:Hard2}, we know that $\pi_2=\{x_4, x_5,x_6\}$ and $\pi_3 = \{x_3\}$. Moreover, due to Remark~\ref{remark} and Observation~\ref{O:notInTau}, we know that $x_2\in \pi_1$ and $x_2\notin \tau_1$. Moreover, since $p = k+6$, we know that $|\pi_1| = k+2$. Let $\pi_1 = S'\cup \{x_2\}$ ($S' \subseteq \{x_1,v_1,\ldots,v_N\}$ and $|S'| = k+1$). Now, we have the following claim. 
    
    \begin{clm}\label{C:L}
     $x_1\in \pi_1$.   
    \end{clm}
    \begin{proofofclaim}
    Targeting contradiction assume that $x_1\notin \pi_1$, i.e., $S' \subset \{v_1,\ldots,$ $v_N\}$. Now consider $\pi_2$. First, observe that $x_4 \notin \tau_2$ since $G[\pi_2]$ have two connected components, each containing exactly one vertex ($x_5$ and $x_6$).
    Therefore, for $\Pi$ to respect \EFO (resp., \EFX), $u_1(\pi_1) \geq u_1(x_4)$. Hence, $u_1(x_2)+u_1(S') \geq 3NC+kC+t$,  i.e., $2NC+u_1(S') \geq 3NC+kC+t$, i.e., $u_1(S') \geq NC+kC+t$. Since $N\geq k+2$, this implies that $u_1(S') \geq (k+3)C+t$. But this is not possible since $u_1(S') \leq (k+1)C +C$ (as $u_1(v_i) = C+a_i$, $|S'| = k+1$, and $\sum a_i = C$). Hence, we reach a contradiction.
    \end{proofofclaim} 

    Hence, due to Claim~\ref{C:L}, we have that $x_1 \in \pi_1$. Now, Let $S = \pi_1\setminus \{x_1,x_2\}$. Moreover, let $B \subseteq \{a_1,\ldots,a_N\}$ such that if $v_i \in S$, then $a_i \in B$. Note that $|S|=|B| = k$. Moreover, observe that due to the construction of $B$ and the definition of $\mathcal{U}$, $u_1(S) = kC+ \sum_{b\in B} b$. Now, we have the following crucial claim.
    
    \begin{clm}\label{C:F}
      $\sum_{b\in B} b = t$.
    \end{clm}
    \begin{proofofclaim}
    First, we will prove that $\sum_{b\in B} b \geq t$. Recall that, as discussed in the proof of Claim~\ref{C:L}, $x_4\notin \tau_2$ and $x_4 \in \pi_2$ (due to Observation~\ref{O:Hard2}). Therefore, $u_1(\pi_1) \geq u_1(x_4) = 3NC+kC+t$, i.e., $u_1(x_1)+u_1(x_2) +u_1(S) \geq 3NC+kC+t$, i.e., $NC+2NC+u_1(S) \geq 3NC+kC+t$. Hence, $u_1(S) \geq kC+t$. Recall that $\sum_{v\in S} u_1(v) = kC+ \sum_{b\in B} b$. Therefore, by combining these two equations, $kC+ \sum_{b\in B} b \geq kC+t$, i.e., $\sum_{b\in B} b \geq t$.

    Second, we will prove that $\sum_{b\in B} b \leq t$. Now, consider $\pi_3$ and recall that $\pi_3 = \{x_3\}$ (due to Observation~\ref{O:Hard2}). Recall that $\{x_1,x_2\} \subseteq \pi_1$. As discussed in the proof of Claim~\ref{C:size1}, for any vertex $v \in \{v_1,\ldots, v_N\}$, $u_3(x_1) > u_3(v)$ and $u'_3(x_1) < u'_3(v)$. Therefore, to satisfy \EFO (resp., \EFX), $u_3(x_3) \geq u_3(\pi_1) - u_3(x_1) = u_3(S) +u_3(x_2)$ (resp. $u'_3(x_3) \geq u'_3(\pi_1) - u'_3(x_1) = u'_3(S) +u'_3(x_2)$).  Since $u_3(x_3) = u'_3(x_3),  u_3(S) =  u'_3(S)$, and $u_3(x_2) = u'_3(x_2)$, it is sufficient to consider the equation $u_3(x_3) \geq u_3(S) +u_3(x_2)$. Hence, $2NC+kC+t \geq u_3(S)+ 2NC$, i.e., $u_3(S) \leq kC+t$. Since $u_3(S) = u_1(S) = kC+ \sum_{b\in B} b$ (as for $i\in [N]$, $u_3(v_i) = u_1(v_i)$), we have that $\sum_{b\in B} b \leq t$. 

    Now, combining $\sum_{b\in B} b \leq t$ and $\sum_{b\in B} b \geq t$, we have that $\sum_{b\in B} b = t$.    
    \end{proofofclaim}
    Finally, since $\sum_{b\in B} b = t$, $|B|=k$, we have that $(a_1,\ldots, a_N, t)$ is a Yes-instance of \ksum.  
\end{proof}
Finally, we have the following theorem as a consequence of our construction, Proposition~\ref{P:ksum}, and Lemma~\ref{LF:EFX} and Lemma~\ref{L:EFX2}.
\EFOHard*

Finally, we have the following remark concerning our hardness results.
\begin{remark}
Theorem~\ref{th:EFHard} and Theorem~\ref{th:EFOHard} establish that, for $\varphi \in \{ \mathsf{EF},\mathsf{EF1},\mathsf{EFX}\}$, $\varphi$-\CFD~is W[1]-hard parameterized by $p+|A|+t$ when the valuations are encoded in binary. We remark that these results cannot be directly generalized to unary valuations as \ksum is in P when $M$ is encoded in unary. 
\end{remark}

\section{Kernelization by $\mathsf{vcn}+p+\mathsf{val}$}
Let $(G,A,\mathcal{U},p)$ be an instance of \CFD. We consider the parameterization of \CFD~parameterized by $\mathsf{vcn}+p+\mathsf{val}$. Recall that $\mathsf{vcn}$ is the vertex cover of the input graph $G$ and $\mathsf{val} = |\{u_i(v)~|~i\in[n],v\in V(G)\}|$. Moreover, let $\mathsf{VAL} = \{u_i(v)~|~i\in[n],v\in V(G)\}$ and $\mathcal{A}$ be the set of all agent types. 

Let $U$ be a vertex cover of size $t$. If no such vertex cover is given, then we can compute a vertex cover $U$ of size $t\leq 2\mathsf{vcn}$ using a polynomial-time approximation algorithm~\cite{bookApprox}. Let $I$ be the independent set $V(G)\setminus U$. First, we provide an overview of our kernelization algorithm. Our kernelization algorithm uses techniques similar to the ones used by Deligkas et al.~\cite{delgikas} to provide an \FPT algorithm for \CFDO parameterized by $\mathsf{vcn}+|\mathcal{A}|+\mathsf{val}$. We partition the vertices of $I$ into equivalence classes such that for any two vertices $u$ and $v$ of the same equivalence class, $N_U(u)= N_U(v)$ and $u_i(v) = u_i(u)$ for each $i\in A$. Since vertices of each equivalence class are ``indistinguishable'' for the agents, we keep only at most $p$ vertices from each equivalence class. We then establish that we can have at most $2^t\mathsf{val}^{|\mathcal{A}|}$ many equivalence classes, thus giving us a kernel with at most $2^t\mathsf{val}^{|\mathcal{A}|}p+t$ vertices. Below, we discuss these ideas formally. 

\subsection{Exponential Kernel for \EF-\CFD, \EFO-\CFD, and \EFX-\CFD}
We have the following reduction rule.
\begin{RR}[RR\ref{RR1}]\label{RR1}
Let $(G,A,\mathcal{U},p)$ be an instance of $\varphi$-\CFD~where $\varphi \in \{ \mathsf{EF},\mathsf{EF1},\mathsf{EFX}\}$. Let $S\subseteq I$ be a set of vertices such that $|S|>p$ and for any two vertices $u,v\in S$, $N_U(u)= N_U(v)$ and for each agent $i\in A$, $u_i(u)=u_i(v)$. Moreover, let $v_1,\ldots v_{|S|}$ be an ordering of vertices in $S$. Then, let $H \Leftarrow G-\set{v_{p+1},\ldots, v_{|S|}}$. 
\end{RR}

We have the following lemma to prove that RR\ref{RR1} is safe.

\begin{lemma}\label{L:RR1}
RR\ref{RR1} is safe.
\end{lemma}
\begin{proof}
Since $H$ is an induced subgraph of $G$, due to Lemma~\ref{L:subgraph}, if $(H,A,\mathcal{U},p)$ is a Yes-instance of $\varphi$-\CFD~($\varphi \in \{ \mathsf{EF},\mathsf{EF1},\mathsf{EFX}\}$), then $(G,A,\mathcal{U},p)$ is a Yes-instance of $\varphi$-\CFD.

Next, we prove that if $(G,A,\mathcal{U},p)$ is a Yes-instance of $\varphi$-\CFD, then $(H,A,\mathcal{U},p)$ is a Yes-instance of $\varphi$-\CFD.
Let $\Pi=(\pi_1, \ldots, \pi_n)$ be an allocation that satisfies $\varphi$ for $(G,A,\mathcal{U},p)$. Let $X \subseteq \set{v_1,\ldots,v_p}$ such that, in allocation $\Pi$, no vertex in $X$ is assigned to any agent $i$. Moreover, let $Y\subseteq \set{v_{p+1},\ldots, v_{|S|}}$ such that each vertex in $Y$ is assigned to some agent $i$ in $\Pi$. Since the total number of vertices that are assigned to agents in $\Pi$ is $p$, observe that $(p-|X|) + |Y| \leq p$, and hence $|X| \geq |Y|$. Let $x_1,\ldots, x_{|X|}$ be an ordering of vertices in $X$ and $y_1,\ldots,y_{|Y|}$ be an ordering of vertices in $Y$. Now, we will create a $\varphi$ allocation $\Pi' = (\pi'_1,\ldots,\pi'_n)$ for $(H,A,\mathcal{U},p)$ using $\Pi$ in the following manner.
For each agent $i\in A$, $\pi'_i = \set{(\pi_i \setminus \set{y_j}) \cup \set{x_j}~|~ y_j\in \pi_i}$. 
Now, we observe some properties of $\Pi'$ to establish that $\Pi'$ is indeed a $\varphi$-\CFD allocation for $(H,A,\mathcal{U},p)$. First, we have the following claim to establish that $\Pi'$ is a valid allocation for $(H,A,\mathcal{U},p)$.

\begin{clm}\label{O:validKernel} 
The following are true for the tuple $\Pi'$. 
\begin{enumerate}
    \item For $i \in A$, $\pi'_i \subseteq V(H)$ and $H[\pi'_i]$ is connected.
    \item For distinct $i,j\in A$, $\pi'_i \cap \pi'_j = \emptyset$.
    \item $|\bigcup_{i\in A} \pi'_i| = p$.
\end{enumerate}
\end{clm}
\begin{proofofclaim}
\begin{enumerate}
    \item First, we show that for $i\in A$, $H[\pi'_i]$ is connected. If $|\pi_i| = 1$, then observe that $|\pi'_i| = 1$ and hence $H[\pi'_i]$ is connected. So, we assume $|\pi_i| >1$. If $\pi'_i = \pi_i$, then clearly $H[\pi'_i]$ is connected since $H$ is an induced subgraph of $G$. Otherwise, there is some vertex $y_j$ such that $y_j\in \pi_i$, $y_j \notin \pi'_i$, and $x_j \in \pi'_i$. Since $G[\pi_i]$ is connected and $y_j \in I$, there is at least one vertex $w\in N[y_j]$ such that $w\in \pi_i$. Note that $w\in \pi'_i$ as well (by definition of $\pi'_i$). Since $N_U(y_j) = N_U(x_j)$, observe that $H[(\pi'_i\setminus \set{y_j}) \cup \set{x_j}]$ is connected. Therefore, $H[\pi'_i]$ is connected.
    
    \item Second, we show that for distinct $i,j\in A$, $\pi'_i\cap \pi'_j = \emptyset$. Since $\pi_i \cap \pi_j = \emptyset$, $(\pi'_i\setminus X)\cup (\pi'_j\setminus X) = \emptyset$. Targeting contradiction, assume that $\pi'_i \cap \pi'_j \neq \emptyset$. Then, there exists some $x_\ell \in X$ such that $x_\ell \in \pi'_i$ and $x_\ell \in \pi'_j$. Therefore, by the construction of $\pi'_i$ and $\pi'_j$, $y_\ell \in \pi_i$ and $y_\ell \in \pi_j$, which is a contradiction to the fact that $\pi_i\cap \pi_j = \emptyset$.
    
    \item Finally, we show that $|\bigcup_{i\in A} \pi_i| = p$. This can easily be seen by observing that for $i\in A$, $|\pi'_i| = |\pi_i|$. (It is trivial when $\pi_i = \pi'_i$. Otherwise, if there is some element $y\in \pi_i\cap Y$, then we replace it with a unique element from $X$ in $\pi'_i$.)
\end{enumerate}
This completes our proof.    
\end{proofofclaim}

Claim~\ref{O:validKernel} establishes that $\Pi'$ is indeed an allocation for $(H,A,\mathcal{U},p)$. In the following claim, we prove additional properties of $\Pi'$ that we will use to establish that $\Pi'$ also satisfies $\varphi$.

\begin{clm}\label{O:goodPi}
Consider the allocations $\Pi$ and $\Pi'$. For a bundle $\pi'_i$, let $\tau'_i = \set{v\in \pi'_i~|~ H[\pi'_i \setminus \set{v}] \text{ is connected}}$. Then,  for $i,j \in A$,
\begin{enumerate}
    \item $u_i(\pi'_j) = u_i(\pi_j)$,
    \item $\max_{v\in \tau'_j} u_i(v) = \max_{w \in \tau_j} u_i(w)$, and 
    \item $\min_{v\in \tau'_j} u_i(v) = \min_{w \in \tau_j} u_i(w)$.
\end{enumerate}
\end{clm}
\begin{proofofclaim}
All these statements are trivial when $\pi'_j = \pi_j$. So, for the rest of the proof, we assume that $\pi'_j \neq \pi_j$.
\begin{enumerate}
    \item Due to the construction of $\pi'_j$, $\pi'_j \setminus (X \cup Y) = \pi_j \setminus (X\cup Y)$. Moreover, note that we replace each element of $\pi_j \cap Y$ with a unique element from $X$ to get $\pi'_j$. Since for all vertices $x\in X$ and $y \in Y$, $u_i(x) = u_i(y)$ (for $i\in A$), we have that $u_i(\pi'_j) = u_i(\pi'_j)$. 
    
    \item Using the arguments used in Case (1), it is easy to see that there is a bijection $f:\pi_j \rightarrow \pi'_j$ such that for $v\in \pi_j$, $u_i(v) = u_i(f(v))$. Therefore, $\max_{v\in \tau'_j} u_i(v) = \max_{w \in \tau_j} u_i(w)$.
    
    \item Similarly to Case (2), due to the bijection $f$, $\min_{v\in \tau'_j} u_i(v) = \min_{w \in \tau_j} u_i(w)$.
\end{enumerate}
This completes our proof.    
\end{proofofclaim}

Finally, we use Claim~\ref{O:goodPi} to show that $\Pi'$ satisfies $\varphi$. For the sake of completeness, we now prove explicitly for each $\varphi \in \{ \mathsf{EF},\mathsf{EF1},\mathsf{EFX}\}$ that if $\Pi'$ is not a $\varphi$ allocation for $(H,A,\mathcal{U},p)$, then $\Pi$ is not a $\varphi$ allocation for $(G,A,\mathcal{U},p)$. Targeting contradiction, assume that $\Pi'$ is not a $\varphi$ allocation for $(H,A,\mathcal{U},p)$. We have the following cases:

\begin{enumerate}
    \item $\mathbf{\varphi}= \textbf{EF}:$ There exist some $i,j\in A$ such that $u_i(\pi'_i) < u_i(\pi'_j)$. Since for each $p,q \in A$, $u_p(\pi_q) = u_p(\pi'_q)$ (due to Claim~\ref{O:goodPi} (1)), we have that $u_i(\pi_i) <u_i(\pi_j)$. This contradicts the fact that $\Pi$ is an \EF~allocation for $(G,A,\mathcal{U},p)$.
    
    \item $\mathbf{\varphi}= \textbf{EF1}:$ There exist some $i,j\in A$ such that $u_i(\pi'_i) < u_i(\pi'_j)-\max_{v \in \tau'_j} u_i(v)$. Since for each $p,q \in A$, $u_p(\pi_q) = u_p(\pi'_q)$ and $\max_{v \in \tau_q} u_p(v) = \max_{w \in \tau'_q} u_p(w)$ (due to Claim~\ref{O:goodPi}), we have that $u_i(\pi_i) <u_i(\pi_j)- \max_{v \in \tau_j} u_i(v)$. This contradicts the fact that $\Pi$ is an \EFO~allocation for $(G,A,\mathcal{U},p)$. 
    
    \item $\mathbf{\varphi}= \textbf{EFX}:$ There exist some $i,j\in A$ such that $u_i(\pi'_i) < u_i(\pi'_j)-\min_{v \in \tau'_j} u_i(v)$. Since for each $p,q \in A$, $u_p(\pi_q) = u_p(\pi'_q)$ and $\min_{v \in \tau_q} u_p(v) = \min_{w \in \tau'_q} u_p(w)$ (due to Claim~\ref{O:goodPi}), we have that $u_i(\pi_i) <u_i(\pi_j)- \min_{v \in \tau_j} u_i(v)$. This contradicts the fact that $\Pi$ is an \EFX~allocation for $(G,A,\mathcal{U},p)$. 
    
\end{enumerate}
This completes the proof of our lemma.
  \end{proof}


In the following lemma we bound the size of our kernel.
\begin{lemma}\label{L:kernelSize}
Let $(G,A,\mathcal{U},p)$ be an instance of $\varphi$-\CFD~such that RR\ref{RR1} cannot be applied to $(G,A,\mathcal{U},p)$. Then, $|V(G)|\leq 2^t\cdot val^{|\mathcal{A}|}\cdot p+t$.
\end{lemma}
\begin{proof}
We say that a subset $S\subseteq I$ is a \textit{equivalence class} if for every $v,w\in S$, $N(w) = N(v)$ and for any two agents $i,j\in A$, $u_i(w) = u_i(v)$.
For two vertices $w,v \in V(G)$, note that $u_i(w) = u_i(v)$ for each $i\in A$ if and only if $u_\mathsf{a}(w) = u_\mathsf{a}(v)$ for each $\mathsf{a} \in \mathcal{A}$.  Therefore, the total number of equivalence classes in $I$ can be at most $2^t val^{|\mathcal{A}|}$ (since the total type of neighbourhoods in $I$ can be at most $2^t$). Moreover, note that if any of the equivalence class $S$ contains more than $p$ vertices, then we can apply RR\ref{RR1}. Therefore, each equivalence class contains at most $p$ vertices. Hence, $|I| \leq 2^t\cdot val^{|\mathcal{A}|}\cdot p$. Since $|U| =t$, we have that $|V(G)|\leq 2^t\cdot val^{|\mathcal{A}|}\cdot p+t$.
  \end{proof}

Since $|\mathcal{A}| \leq \mathsf{val}$,  RR\ref{RR1}, along with  Lemma~\ref{L:RR1} and Lemma~\ref{L:kernelSize}, imply that for $\varphi \in \{ \mathsf{EF},\mathsf{EF1},\mathsf{EFX}\}$, $\varphi$-\CFD admits a kernel with at most $p2^{\mathsf{vcn}}\mathsf{val}^{\mathsf{val}}+\mathsf{vcn}$ vertices. This gives us the desired kernel for $\varphi$-\CFD, for $\varphi \in \{ \mathsf{EF},\mathsf{EF1},\mathsf{EFX}\}$, from Theorem~\ref{th:Kernel}.


\subsection{Exponential Kernel for \PROP-\CFD}
In this section, we design an exponential kernel for \PROP-\CFD. First, we give an overview of the ideas leading to this kernelization. As remarked earlier in Section~\ref{SS:P}, since it might so happen that \PROP-\CFD~exists for a subgraph $H$ of $G$ but not for $G$, we cannot simply delete vertices as before (which might decrease the overall valuation of the graph for some agents) to get our kernel. So, we have to keep track of the ``valuation'' that is lost for each agent while deleting the vertices. In this quest, we augment the graph with some dummy vertices and assign them the valuation for each deleted vertex. Hence, even if we start with a unary (resp. binary) valuation instance, we might end up with an instance that does not admit unary (resp. binary) valuations. We tackle this problem by establishing that if the valuation for a dummy vertex gets ``too large'', then we are dealing with a No-instance. This ensures that we have a kernel that respects the unary (resp. binary) valuation if the original instance respects unary (resp. binary) valuation. Now, we discuss the details of our kernelization algorithm.

Let $(G,A,\mathcal{U},p)$ be an instance of \PROP-\CFD. Recall that $U$ is a vertex cover of $G$ and $I= V(G)\setminus U$. First, we use the following preprocessing step to generate an augmented instance $(G',A',$ $\mathcal{U}',p')$ in the following manner. 

\medskip
\noindent\textsc{Preprocessing Step}: We get $(G',A',\mathcal{U}',p')$ by adding $n$ dummy agents, $n$ dummy vertices, and setting $p' = p+n$. More formally:
\begin{enumerate}
    \item Let $V(G') = V(G)\cup\set{d_1,\ldots, d_n}$. Fix a vertex $w\in U$. Now, let $E(G) = E(G)\cup \set{d_id_j~|~i,j\in A} \cup \set{d_1w}$. 
    \item $A' = A\cup \set{n+1,\ldots, 2n}$.
    \item For $i\in [n]$ and for $v\in V(G)$, let $u'_i(v) = 2 u_i(v)$. Moreover, for $i\in [n]$ and for $v\in \{d_1,\ldots, d_n\}$, let $u'_i(v) = 0$ and $u'_{n+i}(d_i) = 1$. Moreover, for $n<i\leq 2n$ and $v\notin \set{d_1,\ldots, d_n}$, let $u'_i(v) = 0$.
    \item $p' = p+n$.
\end{enumerate}

Now, we have the following lemma.
\begin{lemma}\label{L:preprocess}
$(G,A,\mathcal{U},p)$ is a Yes-instance of \PROP-\CFD~if and only if $(G',A',\mathcal{U}',p')$ is a Yes-instance of \PROP-\CFD.
\end{lemma}
\begin{proof}
In one direction, suppose $(G,A,\mathcal{U},p)$ is a Yes-instance of \PROP-\CFD~with $\Pi = (\pi_1,\ldots,\pi_n)$ as a \PROP~allocation. Then, observe that $\Pi' = (\pi'_1,\ldots, \pi'_{2n})$ where for $i\in [n]$, $\pi'_i = \pi_i$ and $\pi'_{n+i} = \set{d_i}$ is a \PROP~allocation for $(G',A',\mathcal{U}',p')$.

In the reverse direction, suppose $(G',A',\mathcal{U}',p')$ is a Yes-instance of \PROP-\CFD~with $\Pi' = (\pi'_1,\ldots, \pi'_{2n})$ as a \PROP~allocation. First, we prove the following claim.

\begin{clm}\label{C:prop}
For $i\in [n]$, $d_i \in \pi'_{n+i}$.
\end{clm}
\begin{proofofclaim}
The proof follows from the fact that for $i\in [n]$,  $u_{n+i}(V(G)\setminus \set{d_i}) = 0$ and $u_{n+i}(d_i) > 0$.
\end{proofofclaim}

Let $q=\sum_{i\in [n]}|\pi'_i|$. Note that $q\leq p$ (since $\sum_{i\in [2n]} |\pi'_i| = p+n$, and $\sum_{n+1\leq i \leq 2n} |\pi'_i| \geq n$). Moreover, note that, due to Claim~\ref{C:prop}, for $i\in [n]$, $\pi'_i \subseteq V(G)$ as well. Next, we claim that $\Pi = (\pi_1,\ldots,\pi_n)$ where $\pi_i = \pi'_i$, for $i\in [n]$, is a \PROP~allocation for $(G,A,\mathcal{U},q)$. 

\begin{clm}\label{C:prop2}
$\Pi$ is a \PROP~allocation for $(G,A,\mathcal{U},q)$.  
\end{clm}
\begin{proofofclaim}
First, observe that, due to the construction of $\Pi$, $\Pi$ is an allocation for $(G,A,\mathcal{U},q)$. For the rest of the proof, we prove  that $\Pi$ satisfies \PROP.
Targeting contradiction, assume that $\Pi$ does not satisfy \PROP. Hence, there is some $i\in [n]$ such that $u_i(\pi_i) < \frac{u_i(V(G))}{n}$. 
Since $u'_i(v)= 2u_i(v)$, for $v\in V(G)$ and $i\in [n]$, and $\pi'_i = \pi_i$, we have that $u'_i(\pi'_i) = 2u_i(\pi_i) < \frac{2u_i(V(G))}{n}$. Since $u'_i(V(G)) = 2u_i(V(G))$, for $i\in [n]$ (due to \textsc{Preprocessing Step}), we have that $u'_i(\pi'_i) < \frac{u'_i(V(G))}{n} < \frac{u'_i(V(G))}{2n}$, a contradiction to the fact that $\Pi'$ is a \PROP allocation for $(G',A',\mathcal{U}',p')$.
  \end{proofofclaim}



Due to Claim~\ref{C:prop2}, $(G,A,\mathcal{U},q)$ (where $q\leq p$) is a Yes-instance of \PROP-\CFD. Therefore, due to Observation~\ref{O:proper}, $(G,A,\mathcal{U},p)$ is also a Yes-instance of \PROP-\CFD.
  \end{proof}

Let $D = \set{d_1,\ldots,d_n}$ be the set of dummy vertices added in the \textsc{Preprocessing Step}. Now, we have the following reduction rule. 

\begin{RR}[RR\ref{RR2}]\label{RR2}
Let $(G',A',\mathcal{U}',p')$ be an instance of \PROP-\CFD~(obtained from an initial instance $(G,A,\mathcal{U},p)$ after applying \textsc{Preprocessing Step} and possibly RR\ref{RR2}). Let $S\subseteq I$ be a set of vertices such that $|S|>p$ and for any two vertices $u,v\in S$, $N_U(u)= N_U(v)$ and for each agent $i\in A$, $u'_i(u)=u'_i(v)$. Moreover, let $v_1,\ldots v_{|S|}$ be an ordering of vertices in $S$. Then, let $G'' \Leftarrow G'-\set{v_{p+1},\ldots, v_{|S|}}$. Moreover, we modify $\mathcal{U}'$ to get $\mathcal{U}''$ in the following manner.  For $i\in A'$ and $v \in V(G'')$, $u''_i(v)\Leftarrow u'_i(v)$ and, finally, for $i\in A$ $(i\in [n])$, $u''_i(d_i)\Leftarrow u''_i(d_i)+\sum_{p+1 \leq j \leq |S|}u'_i(v_j)$. 
\end{RR}

We have the following lemma to prove that RR\ref{RR2} is safe.
\begin{lemma}\label{L:RR2}
RR\ref{RR2} is safe.
\end{lemma}
\begin{proof}[Proof Sketch]
The proof is based on techniques similar to the proof of Lemma~\ref{L:RR1}. So,  we present here a sketch of the proof.

We first claim that for $i\in A'$, $u'_i(V(G')) = u''_i(V(G''))$. Due to the construction of $G''$, observe that if there is a vertex $v \in V(G')\setminus V(G'')$, then we add $u'_j(v)$ to $u'_j(d_j)$ (for each $j\in [n]$). Hence, $u'_i(V(G')) = u''_i(V(G''))$. Moreover, due to Claim~\ref{C:prop}, in any \PROP~allocation for $(G',A',\mathcal{U}',p')$ or $(G'',A',\mathcal{U}'',p')$, for $i\in[n]$, the vertex $d_i$ is assigned to the agent $n+i$. 

In one direction, suppose $(G'',A',\mathcal{U}'',p')$ is a Yes-instance and let $\Pi'' = (\pi''_1,\ldots, \pi''_{2n})$ be a \PROP~allocation for $(G'',A',\mathcal{U}'',p')$. Let $\Pi'=\Pi''$. Since $G''$ is an induced subgraph of $G'$, observe that $\Pi'$ is an allocation for $G'$. Moreover, since for each agent $i\in [2n]$, $u'_i(V(G')) = u''_i(V(G''))$ and $\pi'_i = \pi''_i$, it follows that $\Pi'$ satisfies \PROP~for $(G',A',\mathcal{U}',p')$ as well.

In the reverse direction, suppose $(G',A',\mathcal{U}',p')$ be a Yes-instance and let $\Pi' = (\pi'_1,\ldots, \pi'_{2n})$
be a \PROP~allocation for $(G',A',\mathcal{U}',p')$. Similarly to the proof of Lemma~\ref{L:RR1}, let $X\subseteq \set{v_1,\ldots,v_p}$ such that no vertex in $X$ is assigned to any agent in $\Pi'$. Moreover, let $Y \subseteq \set{v_{p+1},\ldots, v_{|S|}}$ such that each vertex in $Y$ is assigned to some agent in $\Pi'$. Similarly to the proof of Lemma~\ref{L:RR1}, here also $|X|\geq |Y|$. Here also, let $\set{x_1,\ldots,x_{|X|}}$ be an ordering of vertices in $X$ and $\set{y_1,\ldots, y_{|Y|}}$ be an ordering of vertices in $Y$. Now, we create an allocation $\Pi''$ for $(G'',A',\mathcal{U}'',p')$ in the following manner. For each agent $i\in A'$, $\pi''_i= \set{(\pi'_i\setminus \set{y_j})\cup \set{x_j}~|~y_j\in \pi'_i}$. Now, using the fact that $u'_i(V(G')) = u''_i(V(G''))$ and that for each agent $i\in A'$, $u'_i(\pi'_i) = u''_i(\pi''_i)$, it follows that $\Pi''$ satisfies \PROP~for $(G'',A',\mathcal{U}'',p')$. 
  \end{proof}

Next, we have the following rule that we apply after we have applied RR\ref{RR2} exhaustively. This rule ensures that the valuation on any vertex does not get ``much larger'' than the size of the kernel.

\begin{RR}[RR~3]\label{RR3}
Let $(G',A',\mathcal{U}',p')$ be an instance we get after applying \text{Preprocessing Step} and then an exhaustive application of RR\ref{RR2}. If for some $i\in [n]$,  $u_i(d_i) > u_i(V(G')\setminus D)$, then report a No-instance. 
\end{RR}

\begin{lemma}\label{L:RR3}
RR\ref{RR3} is safe.
\end{lemma}
\begin{proof}
Targeting contradiction, assume that for some $i\in [n]$, $u_i(d_i) > u_i(V(G')\setminus D)$ and $(G',A',\mathcal{U}',p')$ is a Yes-instance and let $\Pi'=(\pi'_1,\ldots,\pi'_{2n})$ is an allocation satisfying \PROP. Note that (due to arguments identical to Claim~\ref{C:prop}) no vertex in $D$ is assigned to the agent $i$ (for $i\in [n]$) in any \PROP~allocation for $(G',A',\mathcal{U}',p')$. Hence, $\pi'_i \subseteq V(G')\setminus D$. Therefore, $u'_i(\pi'_i) \leq u'_i(V(G')\setminus D) < \frac{u'_i(V(G'))}{2}$. Thus, $u'_i(\pi'_i) < \frac{u'_i(V(G'))}{2n}$ (for any $n\geq 1$), contradicting the fact that $\Pi'=(\pi'_1,\ldots,\pi'_{2n})$ is an allocation for $(G',A',\mathcal{U}',p')$ satisfying \PROP.   
  \end{proof}

Next, we establish a bound on the size of the reduced instance. We have the following lemma.

\begin{lemma}\label{L:sizeProp}
Let $(G,A,\mathcal{U},p)$ be an instance of \PROP-\CFD~such that $G$ has a vertex cover $U$ of size $t$. Let $(G',A',\mathcal{U}',p')$ be the instance we get after applying the \textsc{Preprocessing Step} and an exhaustive application of RR\ref{RR2}. Then, $|V(G')| \leq 2^t val^{|\mathcal{A}|}p+t+n$ and $p' \leq p+n$.
\end{lemma}
\begin{proof}
The proof follows from the fact that we add $n$ vertices and agents in the \textsc{Preprocessing step} and then using the identical arguments as used in the proof of Lemma~\ref{L:kernelSize}.
  \end{proof}

Since $|A|= n \leq p$ and $|\mathcal{A}| \leq val$, RR\ref{RR2}, along with Lemmas~\ref{L:preprocess}, \ref{L:RR2}, and \ref{L:sizeProp}, imply that \PROP-\CFD~admits a kernel with at most $2^t val^{val}p+t+p$ vertices. Moreover, we have the following remark.


\begin{remark}
Due to RR\ref{RR3} and Lemma~\ref{L:RR3}, if the initial instance has unary valuation, then the kernel also has unary valuation. Similarly, if the initial instance has binary valuation, then the kernel also has binary valuations.
\end{remark}
Finally, due to kernelization algorithms presented in this section, we have the following theorem.
\Kernel*

\section{Incompressibility}\label{S:NPoly}
In this section, we complement our exponential kernels by showing that it is unlikely that $\varphi$-\CFD~admits a polynomial kernel parameterized by $\mathsf{vcn}+val+p+|A|$, where $\varphi \in \{ \mathsf{PROP}, \mathsf{EF},\mathsf{EF1},\mathsf{EFX}\}$. For this purpose, we first define the following problem. In \RBDS, we are given a bipartite graph $G$ with a vertex bipartition $V(G) = T \cup N$ and a non-negative integer $k$. A set of vertices $N'\subseteq N$ is said to be an \textit{RBDS} if each vertex in $T$ has a neighbour in $N'$. The aim of \RBDS is to decide whether there exists an \textit{RBDS} of size at most $k$ in $G$. We assume that $k<|N|$ as otherwise, the problem becomes trivial.
Dom et al.~\cite{RBDSIncompressibility} proved that it is unlikely for \RBDS parameterized by $|T|+k$ to admit a polynomial compression:
\begin{proposition}[\cite{RBDSIncompressibility}]\label{R:RBDS}
\RBDS parameterized by $|T|+k$ does not admit a polynomial compression, unless \NPoly.
\end{proposition}

\subsection{Incompressibility for \PROP and \EF} 
In this section, we show that \PROP-\CFD and \EF-\CFD~parameterized by $\mathsf{vcn}+p+val+|A|$ does not admit a polynomial compression by developing a polynomial parameter transformation from  \RBDS parameterized by $|T|+k$ to \CFD~parameterized by  $\mathsf{vcn}+val+p+|A|$. 

\medskip
\noindent\textbf{Polynomial Parameter Transformation.} Suppose $((G,k),|T|+k)$ is an instance of \RBDS parameterized by $|T|+k$ such that $V(G) = T\cup N$. For brevity, let $P$ denote the value of the parameter $\mathsf{vcn}+\mathsf{val}+p+|A|$. First we construct a graph $G'$ with $V(G') = T'\cup N'$ from $G$ by adding two dummy vertices $d_1$ and $d_2$ such that $T' = T\cup \{d_1,d_2\}$, $N'=N$, and $E(G') = E(G)\cup \{d_1x,d_2x~|~x\in N'\} $. Now, we construct the instance $((G',A,\mathcal{U},p),P)$ of \CFD~such that $((G,k),|T|+k)$ is a Yes-instance if and only if $((G',A,\mathcal{U},p),P)$ is a Yes-instance and $|P| \leq h(|T|+k)$ where $h$ is a polynomial function. Let $p = |T|+k+2$ and $A = \set{1,2}$. We will define the valuation function $\mathcal{U}$ based on the fairness criterion. 

\subsubsection{\PROP-\CFD}For each vertex $v\in N'$, let $u_1(v)=u_2(v) = 0$. For each vertex $v\in T'\setminus \set{d_1}$, let $u_1(v)=u_2(v)=1$ and $u_1(d_1)=u_2(d_1)= |T|+1$. Note that $\mathsf{val}=3$ and $|\mathcal{A}| = 1$ (i.e., identical valuation). Observe that our reduction is a polynomial parameter transformation since $P=\mathsf{vcn}+\mathsf{val}+p+|A| \leq |T'|+3+|T|+k+2+2 = 2|T|+k+9$ (as $T'$ is a vertex cover of $G'$ and $|T'| = |T|+2$). Hence, the following lemma implies the incompressibility of \PROP-\CFD.
\begin{lemma}\label{L:IPROP}
$((G',A,\mathcal{U},p),P)$ is a Yes-instance of \PROP-\CFD if and only if $((G,k),t+k)$ is a Yes-instance of \RBDS.
\end{lemma}
\begin{proof}
    In one direction, let $G$ has an \textit{RBDS} $S$ of size $k$. Then, let $\pi_1 = \set{d_1}$ and $\pi_2 = (T'\setminus \set{d_1})\cup S$.  Observe that both $\pi_1$ and $\pi_2$ (since for any two vertices $x,y \in \pi_2\cap S$, $xd_2,yd_2\in E(G')$ and $S$ is an \textit{RBDS} in $G$) induces connected subgraphs of $G'$. Moreover, $u_1(\pi_1) = u_2(\pi_2) = |T|+1$. Since $u_1(V(G')) = u_2(V(G'))= 2|T|+2$, $\Pi= (\pi_1,\pi_2)$ is a \PROP~allocation for $((G',A,\mathcal{U},p),P)$. Thus, $((G',A,\mathcal{U},p),P)$ is a Yes-instance of \PROP-\CFD.

In the reverse direction, let $((G',A,\mathcal{U},p),P)$ be a Yes-instance of \PROP-\CFD. Then, since $u_1(V(G'))=u_2(V(G')) = 2|T|+2$ and both agents have identical valuation, for \PROP-\CFD, both $u_i(\pi_1)$ and $u_i(\pi_2)$ should be equal to $|T|+1$. The only way to achieve this is to assign the vertex $d_1$ (possibly along with some vertices from $N'$) to one of the agents and the vertices in $T'\setminus \set{d_1}$ (along with some vertices from $N'$) to the other agent. Without loss of generality, assume that $d_1\in \pi_1$ and $T'\setminus \{d_1\}\subseteq \pi_2$. Observe that $|\pi_1| \geq 1$. Therefore, $|\pi_2| \leq p-1 = |T|+k+1$. Since $T'\setminus \{d_1\}\subseteq \pi_2$, $|\pi_2 \cap N'| \leq k$. Moreover, since $G'[\pi_2]$ is connected (and $T'\cup N'$ is the bipartition of $G')$, each vertex in $T'\setminus \set{d_1}$ shares an edge with at least one vertex in $\pi_2 \cap N'$. Hence, $\pi_2 \cap N'$ dominates each vertex in $T'\setminus \set{d_1}$. Therefore, $\pi_2 \cap N'$ is an \textit{RBDS} of $G$ containing at most $k$ vertices (as $T \subseteq T'\setminus \set{d_1}$ and $N = N'$).
\end{proof}

\subsubsection{\EF-\CFD}
 Here, we have the following two cases depending on whether $|T|+k+2$ is odd or even.

\smallskip
\noindent\textbf{Case 1:} $|T|+k+2$ is odd. In this case, for each vertex $v\in N'$, let $u_1(v)=u_2(v) = 5|T|$. For each vertex $v\in T'\setminus \set{d_1}$, let $u_1(v)=u_2(v)=1$ and $u_1(d_1)=u_2(d_1)= 5k|T|+|T|+1$. Note that $\mathsf{val}=3$ and $|\mathcal{A}| = 1$ (i.e., identical valuation). Moreover, similarly to the above case for \PROP-\CFD, here also $P\leq 2|T|+k+9$. Now, we have the following lemma.

\begin{lemma}\label{L:IEFOdd}
   $((G',A,\mathcal{U},p),P)$ is a Yes-instance of \EF-\CFD if and only if $((G,k),|T|+k)$ is a Yes-instance of \RBDS.
\end{lemma}
\begin{proof}   
    In one direction, let $S$ be an \textit{RBDS} in $G$ with $k$ vertices. Then, let $\pi_1 = \set{d_1}$ and $\pi_2= (T'\setminus \set{d_1}) \cup S$. Observe that $|\pi_1|+|\pi_2| = p$, and both $\pi_1$ and $\pi_2$ induce connected subgraphs of $G'$. Moreover, $u(\pi_1) = u(\pi_2) = 5k|T|+|T|+1$, and therefore, none of the agents is envious of the other. Thus, $((G',A,\mathcal{U},p),P)$ is a Yes-instance of \EF-\CFD.
    
    In the reverse direction, let $((G',A,\mathcal{U},p),P)$ be a Yes-instance of \EF-\CFD~and $\Pi=(\pi_1,\pi_2)$ be an \EF~allocation for this instance. First, we show that the vertex $d_1$ is assigned to one of the agents (i.e., either $d_1\in \pi_1$ or $d_1 \in \pi_2$). Targeting contradiction, assume that $d_1 \notin \pi_1 \cup \pi_2$. Since $|\pi_1|+|\pi_2|= |T|+k+2$ is an odd number, $|\pi_1|\neq |\pi_2|$. As the valuations are identical, due to Observation~\ref{O:identical}, \EF-\CFD~implies $u(\pi_1)= u(\pi_2)$. Let $|\pi_1\cap N'| = \alpha$ and $|\pi_2\cap N'| = \beta$. Then, $u(\pi_1) = u(\pi_1\cap N')+u(\pi_1\cap T')= 5\alpha |T|+ (|\pi_1|-\alpha) \leq 5\alpha |T| + |T|+1$. Similarly, $u(\pi_2)= 5\beta |T|+ (|\pi_2|-\beta)\leq 5\beta |T|+|T|+1$. We have the following two cases:
    \begin{enumerate}
        \item $\mathbf{\alpha \neq \beta:}$ Without loss of generality assume that $\alpha > \beta$. Then, $u(\pi_1)-u(\pi_2) \geq 5\alpha|T|-5\beta |T|-|T|-1 \geq 5|T|-|T|-1 \geq 4|T|-1 > 0$ (for $|T|>0$). Hence $u(\pi_1) \neq u(\pi_2)$, which contradicts the fact that $\Pi$ is an \EF~allocation (due to Observation~\ref{O:identical}).
        
        \item $\mathbb{\alpha = \beta:}$ Since $|\pi_1| \neq |\pi_2|$ and $\alpha = \beta$, observe that $|\pi_1|-\alpha \neq |\pi_2|-\beta$. Hence, $u(\pi_1) - u(\pi_2) = 5\alpha |T| + (|\pi_1|-\alpha) - 5\beta |T| -(|\pi_2|-\beta) = (|\pi_1|-\alpha) - (|\pi_2|-\beta) \neq 0$. Hence $u(\pi_1) \neq u(\pi_2)$, which contradicts the fact that $\Pi$ is an  \EF~allocation (due to Observation~\ref{O:identical}).
    \end{enumerate}
    
    Hence, the vertex $d_1$ must be assigned to one of the agents. Without loss of generality, assume $d_1$ is assigned to $\pi_1$. Thus, $u(\pi_1)\geq 5k|T|+|T|+1$. Moreover, let $\pi_1$ be assigned $\alpha$ vertices from $N'$ and $\beta$ vertices from $T'\setminus \set{d_1}$. ( We will argue that $\alpha = \beta = 0$.) Then, $u(\pi_1) = 5k|T|+|T|+1 +5\alpha |T|+\beta = 5(k+\alpha)|T|+|T|+\beta +1$. Now, let $\pi_2$ is assigned $\alpha'$ vertices from $N'$ and $\beta'$ vertices from $T'\setminus \set{d_1}$. Therefore, $u(\pi_2)= 5\alpha'|T|+\beta'$ where $\beta' \leq |T|+1$. We again have the following two cases depending on whether $\alpha'$ equals $k+\alpha$ or not.  
    \begin{enumerate}
        \item $\alpha'= k+\alpha$. In this case, note that $u(\pi_1) = u(\pi_2)$ if and only if $|T|+\beta+1 = \beta'$. It is only possible if $\beta = 0$ and $\beta' = |T|+1$ (since $\beta' \leq |T|+1$).
        
        \item $\alpha' \neq k+\alpha$. We will show that, in this case, $u(\pi_1) \neq u(\pi_2)$. Note that $u(\pi_1)-u(\pi_2) = 5(k+\alpha)|T|+|T|+\beta +1 - 5\alpha'|T|-\beta' = 5(k+\alpha-\alpha')|T|+|T|+1+\beta-\beta'$. Since $\alpha' \neq k+\alpha$, $k+\alpha- \alpha' \neq 0$, and hence, $-1 \geq k+\alpha- \alpha' \geq 1$. Therefore, $-5|T|+|T|+1+\beta-\beta' \geq u(\pi_1)-u(\pi_2) \geq 5|T|+|T|+1+\beta-\beta'$, i.e., $-4|T|+1+\beta-\beta' \geq u(\pi_1)-u(\pi_2) \geq 6|T|+1+\beta-\beta'$. Since $0\leq \beta \leq |T|+1$ and $0 \leq \beta' \leq |T|+1$, we have that $-|T| \leq 1+\beta-\beta' \leq |T|+2$. Therefore, $u(\pi_1) \neq u(\pi_2)$.
    \end{enumerate}
    
    Hence, if $u(\pi_1) = u(\pi_2)$, then $\alpha' = \alpha +k$, $\beta = 0$, and $\beta' = |T|+1$. Moreover, note that $1+\alpha + \beta +\alpha' +\beta' = p =|T| +k+1$. Therefore, $1+\alpha+0 +\alpha'+|T|+1 = |T|+k+1$, i.e, $\alpha+\alpha' = k$. Solving equations $\alpha' = \alpha +k$ and $\alpha+\alpha' = k$ gives us $\alpha' = k$ and $\alpha = 0$. 
    
    Therefore, $\pi_1 = \set{d_1}$ and $\pi_2 = (T'\setminus \set{d_1}) \cup S$ such that $S\subseteq N'$ and $|S| = k$. Moreover, since $G'[\pi_2]$ is connected (and $T'\cup N'$ is a bipartition of $G'$), each element in $T'\setminus \set{d_1}$ has a neighbour in $S$. Since $T\subseteq T'\setminus \set{d_1}$ and $N =N'$, $S$ is an \textit{RBDS} of $G$ containing at most $k$ vertices. Therefore, $((G,k),|T|+k)$ is a Yes-instance.
 \end{proof}   
 
\smallskip  
    \noindent\textbf{Case 2:} $|T|+k+2$ is even. This case is similar to Case 1 (where $|T|+k+2$ is odd). Hence, we only provide the graph construction and the valuations here. Here, we get graph $G''$ by adding one more dummy vertex $d_3$ (i.e., $N''= N'\cup \set{d_3}$ and $T'' = T'$) and the edge $d_1d_3$ (i.e., $E(G'') = E(G') \cup \set{d_1d_3}$). Moreover, we set $u(d_1) = 5k|T|+k$, $u(d_3) = 1$, and $p = |T|+k+3$ ($p$ is odd now). Observe that if $d_1$ is not assigned to any agent, then $d_3$ also cannot be assigned to any agent (since if some agent is assigned only vertex $d_3$, then the other agent can only be assigned one vertex from $T'\setminus \set{d_1}$; but $p \geq 3$). Hence, using the arguments given in Case 1 (as $p$ is odd here as well), we can first establish that \EF-\CFD~does not exist if $d_1$ is not assigned to any agent. Next, using arguments similar to Case 1, we can also establish that if \EF-\CFD~exists in $G''$, then either $\pi_1 = \set{d_1,d_3}$ and $\pi_2 = (T''\setminus \set{d_1}) \cup S$, or vice versa, where $S\subseteq N''\setminus \set{d_3}$ and $|S| = k$. Moreover, we can then establish that $S$ is an \textit{RBDS} of $G$. It is easy to see that if $G$ has an \textit{RBDS} of size at most $k$, then $G''$ admits \EF-\CFD. This establishes the following lemma.
    
    \begin{lemma}\label{L:IEFEven}
    $((G'',A,\mathcal{U},p),P)$ is a Yes-instance of \EF-\CFD~if and only if $((G,k),|T|+k)$ is a Yes-instance of \RBDS.
 \end{lemma}

\subsection{Incompressibility for \EFX and \EFO}
In this section, we prove that, unless \NPoly, \EFX and \EFO do not admit a polynomial compression when parameterized by $\mathsf{vcn}+\mathsf{val}+|A|+p$. Here also, we provide a polynomial parameter transformation from \RBDS parameterized by $|T|+k$ to \EFO (resp., \EFX) parameterized by $P = \mathsf{vcn}+\mathsf{val}+|A|+p$. First, we provide our construction. Let $t = |T|$. Moreover, we assume that $k>1$ and $t>1$.

\begin{figure}
    \centering
    \includegraphics[scale=0.8]{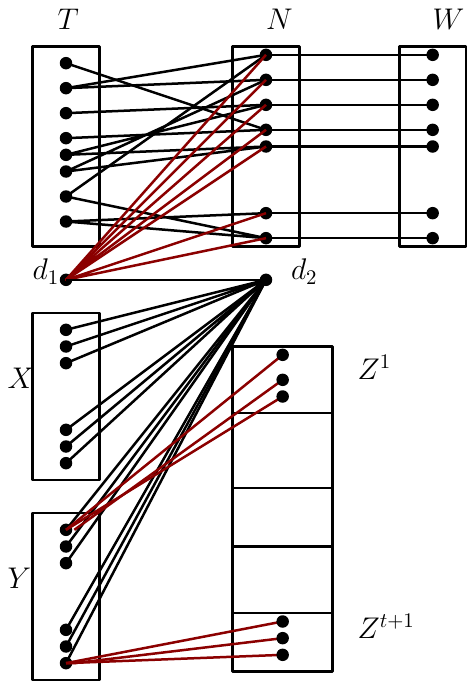}
    \caption{Illustration for the construction of $G'$.}
    \label{fig:construction1}
\end{figure}

\medskip
\noindent\textbf{Construction.} Suppose $((G,k),t+k)$ is an instance of \RBDS parameterized by $t+k$ such that $V(G) = T\cup N$. See Figure~\ref{fig:construction1} for a reference. First, we construct a graph $G'$ with $V(G') = T'\cup N'$ from $G$ by adding some dummy vertices. First, we add two dummy vertices $d_1$ and $d_2$ such that $T' = T\cup \{d_1\}$, $N'=N \cup \{d_2\}$, and $E(G') = E(G)\cup \{d_1x~|~x\in N'\}$ ($d_1d_2$ is also an edge). Next, we add two sets of vertices $X=\{x_1, \ldots, x_{t+3}\}$ and $Y=\{y_1,\ldots,y_{t+1}\}$ of size $t+3$ and $t+1$, respectively, to $T'$ and attach each $x_i$ and $y_j$ to $d_2$, i.e., $T' = T' \cup \set{x_1,\ldots, x_{t+3}} \cup \set{y_1,\ldots,y_{t+1}}$ and $E(G') = E(G')\cup \set{x_id_2, y_jd_2~|~i\in [t+3], j\in [t+1]}$. Next, for each $y_j$, $j\in [t+1]$, we add $32kt$ (here, we want a large enough number) vertices $Z^j=\{z^j_1,\ldots, z^j_{32kt}\}$ to $N'$ and for each $z^j_\ell$, $\ell \in [32kt]$, we add the edge $y_jz^j_\ell$ to $E(G')$. Finally, let $\{v_1,\ldots,v_{|N|}\}$ be an ordering of vertices of $N$. For each $v_i$, $i\in [|N|]$, we  add a vertex $w_i$ and an edge $v_iw_i$ to $G'$. Let $W$ denote the set of vertices $\set{w_1,\ldots, w_{|N}}$. This completes our construction. Note that $|V(G')| = |V(G)|+|N|+ 2+ t+3 +t+1+ (t+1)32kt$.

Next, we set $A= [2t+6]$ and $p = 4k(2t+6) = 8kt+24k$. We define the valuation function $\mathcal{U}$ below. Moreover, for ease of presentation, let us denote the set $T\cup N\cup W \cup \{d_1\}$ by $B'$ and the set $V(G')\setminus (B'\cup \{d_2\})$ by $B$.

\medskip
\noindent\textbf{Valuation Function.} We have two types of agents: Type-I and Type-II, each with cardinality $t+3$, such that agents of the same type have identical valuations. Recall that $|A|= 2t+6$. Hence, without loss of generality, assume that $\{1,\ldots, t+3\}$ are Type-I agents and $\set{t+4,\ldots, 2t+6}$ are Type-II agents. For ease of presentation, we will use $u_\text{I}$ and $u_{\text{II}}$ to represent the valuation function for all Type-I and Type-II agents, respectively. The valuation function $\mathcal{U}$ is defined as follows:
\begin{itemize}
\item for $v\in T$, $u_{\text{I}}(v) = u_{\text{II}}(v) = 1$,
\item $u_{\text{I}}(d_1) = u_{\text{II}}(d_1) = 1$,
\item for $v\in N$, $u_{\text{I}}(v) =0$ and $u_{\text{II}}(v) = 5kt$,
\item for $v\in W$, $u_{\text{I}}(v) =5kt$ and $u_{\text{II}}(v) = 0$,
\item $u_{\text{I}}(d_2) = u_{\text{II}}(d_2) = 10kt$,
\item for $i\in [t+3]$, $u_{\text{I}}(x_i) = 0$ and $u_{\text{II}}(x_i) = 5kt+t+1$,
\item for $i\in [t+1]$, $u_{\text{I}}(y_i) = 5kt+t+1$ and $u_{\text{II}}(y_i) = 0$, and
\item for $i\in [t+1]$, $j\in [32kt]$, $u_{\text{I}}(z^i_j) =  u_{\text{II}}(z^i_j) = 0$.
\end{itemize}

First, we have the following easy observation concerning our construction and $\mathcal{U}$.
\begin{observation}\label{O:sizeOfB}
For each connected component $C$ in $G[B]$, $u_{\I}(C) \leq 5kt+t+1$ and $u_{\II}(C) \leq 5kt+t+1$.
\end{observation}

Now, we will show that $((G',A,\mathcal{U},p),P)$ is Yes-instance of \EFO-\CFD (resp., \EFX-\CFD) if and only if $((G,k),t+k)$ is a Yes-instance of \RBDS. First, we prove one side of this argument in the following lemma.
\begin{lemma}\label{L:OneSide}
If $((G,k),t+k)$ is a Yes-instance of \RBDS, then $((G',A,\mathcal{U},$ $p),P)$ is Yes-instance of \EFO-\CFD as well as of \EFX-\CFD.
\end{lemma}
\begin{proof}
Let $S'$ be an RBDS of $G$ of size at most $k$. Then, let  $S\subseteq N$ be an RBDS of $G$ of size $k$ by possibly adding $k-|S'|$ vertices from $N\setminus S'$ to $S'$. Moreover, let $W'\subseteq W$ such that $|W'| = |S|$ and each vertex in $S$ has a unique neighbour in $W'$ (there is a perfect matching from $S$ to $W'$). Now, we define the allocation $\Pi = (\pi_1, \ldots, \pi_{2t+6})$ in the following manner. First, let $\pi_1 = S\cup W' \cup T \cup d_1$. Observe that $G'[\pi_1]$ is connected since $d_1$ is connected to each vertex in $N$. Moreover, $u_{\text{I}}(\pi_1) = u_{\text{II}}(\pi_1) = 5kt+t+1$ and $|\pi_1| = 2k+t+1$. Next, let $\pi_2 = \{d_2\}$.

Now, let $\ell =32kt-2k-3t-5$. Note that $\ell>0$. Next, let $\pi_3 = \{y_1\} \cup \{z^1_1,\ldots, z^1_{\ell}\}$. Note that $u_{\text{I}} (\pi_3) =  u_{\text{I}}(y_1) = 5kt+t+1$ and $u_{\text{II}} (\pi_3)=0$.  
Next, for $4 \leq i \leq t+3$, let $\pi_i = \{y_{i-2}\}$. Note that for $4 \leq i \leq t+3$, $u_{\I}(\pi_i) = 5kt+t+1$, $u_{\II} (\pi_i) = 0$, and $|\pi_i| = 1$.
Finally, for $t+4 \leq i \leq 2t+6$, (Type-II agents), let $\pi_i = x_{i-t-3}$.
Observe that for $t+4\leq i \leq 2t+6$, $u_{\II}(\pi_i) = 5kt+t+1$ and $|\pi_i| = 1$. 

Now, notice that for agent $i\neq 2$, $u_2(\pi_2) > u_2(\pi_i)$. Hence, agent 2 is not envious of any other agent. Conversely, since $|\pi_2| = 1$, no agent is envious of agent 2 up to one item (resp., up to any item). Finally, observe that for every two agents $i,j$ such that $i\neq 2$ and $j\neq 2$, $u_i(\pi) = 5kt+t+1$ and $u_i(\pi_j) \leq 5kt+t+1$. Hence $\Pi$ satisfies both \EFX and \EFO. Now, the only thing remaining to observe is that $\sum_i|\pi_i| = |\pi_1|+|\pi_2|+|\pi_3| +2t+3 = 2k+t+1+1+32kt-2k-3t-5+2t+3 = 32kt = p$. This completes our proof.
\end{proof}

In the remaining part of this section, we will prove the other side of our argument. For the rest of this section, let $((G',A,\mathcal{U},p),P)$ be Yes-instance of \EFO-\CFD (resp., \EFX-\CFD) and $\Pi = (\pi_1,\ldots, \pi_{2t+6})$ be an allocation for $((G',A,\mathcal{U},p),P)$ that satisfies \EFO (resp., \EFX). We will need the following observations about the properties of $\Pi$. These observations will be valid for both \EFO as well as \EFX. 

\medskip
\noindent\textbf{Important Observations.}
First, we have the following easy observation. 
\begin{observation}\label{O:EFXtrivial}
Let there be an $i$ such that for any vertex $v\in \tau_i$, $u_{\I}(\pi_i)-u_{\I}(v) \geq 5t+1$ (resp., $u_{\II}(\pi_i)-u_{\II}(v) \geq 5t+1$), then at most $t+1$ agents of Type-I (resp., Type-II) can be assigned vertices only from $B' = T\cup N \cup W\cup \{d_1\}$. 
\end{observation}
\begin{proof}
We will prove this only for Type-I agents as the proof for Type-II agents is identical. To satisfy \EFO (resp., \EFX), for each Type-I agent $j$, $u_{\I}(\pi_j) \geq 5t+1$. Moreover, it follows directly from the construction that for each connected component, say $C$, in $N\cup W$, $u_{\I}(C) \leq 5t$ and $u_{\II}(C)\leq 5t$. Therefore, each agent that is assigned vertices only from $T\cup N \cup W\cup \{d_1\}$ is assigned at least one vertex from $T\cup \{d_1\}$. Since $|T\cup \{d_1\}| = t+1$, the  proof follows. 
\end{proof}

Next, we have the following crucial  observation concerning $\Pi$.

\begin{observation}\label{O:2inTau}
If $d_2 \in \pi_i$ for some $i\in A$, then $d_2\in \tau_i$.
\end{observation}
\begin{proof}
Targeting contradiction, let us assume that there is some $i \in A$ such that $d_2 \in \pi_i$ and $d_2 \notin \tau_i$. 
Then, for every $j\in A$, it should be that $u_j(\pi_j) \geq u_j(d_2) = 10kt$. 
Now, observe that for each $j\neq i$, either $\pi_j \subseteq B'= T \cup N \cup W \cup \{d_1\}$ or $\pi_j \subseteq B=V(G') \setminus (A \cup \{d_1,d_2\})$. Since $u_{\I}(\pi_i) \geq 10kt > 5t+1$ and $u_{\II}(\pi_i) \geq 10kt > 5t+1$, due to Observation~\ref{O:EFXtrivial}, at most $t+1$ agents of Type-I as well as of Type-II are assigned vertices from $B$. Since there are $t+3$ agents of each type (Type-I and Type-II), there is at least one agent of each type that is assigned vertices only from $B$. Let $j$ and $j'$ be Type-I and Type-II agents, respectively, that are assigned vertices only from  $B$.

Now, it follows from Observation~\ref{O:sizeOfB} that $u_{\I}(\pi_j) \leq 5kt+t+1 < 10kt$ and $u_{\II}(\pi_{j'}) \leq 5kt+t+1 < 10kt$, which contradicts the fact that $\Pi$ is an \EFO (resp., \EFX) allocation. 
\end{proof}

In the following observation, we establish a necessary condition for $\Pi$ to be an \EFO (resp., \EFX) allocation.

\begin{observation}\label{O:notEFX}
If for every vertex $v\in \tau_i$, $\pi_i\setminus \{v\}$ contains at least $k+1$ vertices from $N$, then $\Pi$ is not an \EFO (resp. \EFX) allocation.
\end{observation}
\begin{proof}
Targeting contradiction, let us assume that  for every vertex $v\in \tau_i$, $\pi_i\setminus \{v\}$ contains at least $k+1$ vertices from $N$ (and $\Pi$ is an \EFO (resp. \EFX) allocation). Since for each vertex $a\in N$, we have that $u_{\II}(a) = 5t$, it follows that for each vertex $v\in \tau_i$, $u_{\II}(\pi_i\setminus \{v\}) \geq 5t(k+1)$. Therefore, due to Observation~\ref{O:PreEFO}, for each Type-II agent $j\in A$, $u_j(\pi_j) \geq 5t(k+1)$. Moreover, since $u_{\II}(\pi_i\setminus \{v\}) \geq 5t(k+1)$, due to Observation~\ref{O:EFXtrivial}, at most $t+1$ Type-II agents can be assigned vertices from $B'$. Moreover, (at most) one agent can be assigned the vertex $d_2$, possibly  along with some other vertices. Therefore, there is a Type-II agent, say $j$, that is assigned vertices only from $B$. Now, due to Observation~\ref{O:sizeOfB}, $u_{\II}(\pi_{j}) \leq 5kt+t+1<5k(t+1)$. This leads to a contradiction.
\end{proof}

The following observation establishes that ``too many'' vertices from $B'$ cannot be allocated to agents in $\Pi$.

\begin{observation}\label{O:notmany}
For $i \in A$, $\pi_i$ can be assigned at most $k$ vertices from $W$ and at most $k+1$ vertices from $N$.     
\end{observation}
\begin{proof}
First, we will show that if $\pi_i$ contains more than $k$ vertices from $W$ or more than $k+2$ vertices from $N$, then for any vertex $v\in \tau_i$, $\pi_i \setminus \{v\}$ contains at least $k+1$ vertices from $N$. We have the following two cases: 
\begin{enumerate}
    \item $\pi_i$ contains  at least $k+1$ vertices from $W$. Since $W$ is an independent set and each vertex in $W$ is uniquely matched to a vertex in $N$, it follows that $\pi_i$ contains at least $k+1$ vertices from $N$ as well (along with some other vertices from the rest of the graph). Let these $k+1$ vertices from $N$ be denoted by $N'$. Observe that for any vertex $v\in N'$, $v\notin \tau_i$ (since removing $v$ from $\pi_i$ disconnects at least one vertex from $W$ that is a part of $\pi_i$).  Therefore, for any vertex $v\in \tau_i$, $\pi_i\setminus \{v\}$ contains at least $k+1$ vertices from $N$. 

    \item $\pi_i$ contains at least $k+2$ vertices from $N$. Then, it follows directly that for any vertex $v\in \pi_i$, $\pi_i \setminus \{v\}$ contains at least $k+1$ vertices from $N$.
\end{enumerate}

The rest of the proof follows from Observation~\ref{O:notEFX}. 
\end{proof}

Our next observation establishes that at least one agent of Type-I is assigned a vertex from $\{y_1,\ldots, y_{t+1}\}$, say $y_i$, along with some vertices (at least two) of the form $z^i_j$. 
\begin{observation}\label{O:min2}
 There is at least one Type-I agent, say $i$, such that $\pi_i$ contains at least one vertex from  $Y$, say $y_j$, such that $y_j\notin \tau_i$. Equivalently, $u_{\I}(\pi_i) - \min_{v\in \tau_i}u_{\I}(v) \geq u_{\I}(\pi_i) - \max_{v\in \tau_i}u_{\I}(v) \geq 5kt+t+1$.
\end{observation}
\begin{proof}
Due to Observation~\ref{O:notmany}, each agent gets at most $k$ vertices from $W$  and at most $k+1$ vertices from $N$. Thus, in total, at most $(2t+6)(2k+1)$ vertices are assigned from $N\cup W$. Therefore, the number of vertices that can be assigned without assigning any vertex from $Z$ is at most $(2t+6)(2k+1)+|T|+|\{d_1,d_2\}|+|X|+|Y| = 4kt+2t+12k+6+t+1+1+(t+3)+(t+1) = 4kt+ 5t+ 12k+12$ vertices. Therefore, at least $p-4kt-5t-12k-12 = 8kt+24k-4kt-5t-12k-12 = 4kt+12k-5t-12$ vertices are assigned from $Z$. Recall that $|A| = 2t+6$. Observe that for $k>1$ and $t>1$, $4kt+12k-5t-12 > 2t+6$, and hence, there is at least one agent, say $i\in A$, that receives two vertices of the form $z^q_r$ and $z^{q'}_{r'}$. Now, we distinguish the following two cases:
\begin{enumerate}
    \item $q = q'$. In this case, observe that the vertex $y_q$ is also assigned to the agent $i$ since $Z$ is an independent set (from construction) and each vertex of the form $z^q_{r}$ is attached to the vertex $y_q$. Hence, in this case, observe that since deleting $y_q$ from $\pi_i$ gives at least two connected components containing vertices $z^q_r$ and $z^q_{r'}$, $y_q$ is the vertex we are looking for.

    \item $q\neq q'$. Observe that here, due to the construction of $G'$, to assign both $z^q_r$ and $z^{q'}_{r'}$ to the agent $i$, we have to assign both $y_q$ and $y_{q'}$ also to $i$. Since $Y$ is an independent set connected by $d_2$, we have to, in turn, assign $d_2$ to agent $i$ as well. Now, in this case, observe that $d_2\in \pi_i$ and $d_2\notin \tau_i$ (since removing $d_2$ from $\pi_i$ gives at least two connected components containing $y_q$ and $y_{q'}$). Notice that this contradicts Observation~\ref{O:2inTau}.
\end{enumerate}
This completes our proof. 
\end{proof}

Finally, we have the following lemma that proves the other side of the reduction.
\begin{lemma}\label{L:EFOreverse}
If $((G',A,\mathcal{U},p),P)$ is Yes-instance of \EFO-\CFD (resp., \EFX-\CFD), then $((G,k),t+k)$ is a Yes-instance of \RBDS.    
\end{lemma}
\begin{proof}
Let $\Pi= (\pi_1,\ldots, \pi_{2t+6})$ be an \EFO (resp., \EFX) allocation for $((G',A,\mathcal{U},p),P)$. Then, due to Observation~\ref{O:min2} and Observation~\ref{O:PreEFO}, we have that for each Type-I agent $i$, $u_i(\pi_i) \geq 5kt+t+1$.

Next, we establish that there is at least one Type-I agent $i$ such that $\pi_i \subseteq B'$. First, observe that there can be at most one (Type-I) agent to whom the vertex $d_2$ is assigned, possibly along with some other vertices. Now, observe that once $d_2$ is allocated to some agent, say $j$, every other agent $j' \neq j$ can be assigned vertices either only from $B'$ or only from $B$, i.e., either $\pi_{j'} \subseteq B'$ or $\pi_{j'} \subseteq B$. Since for every vertex $v\in B \setminus Y$, $u_{\I}(v) = 0$, each Type-I agent that is assigned vertices only from $B$ is assigned at least one vertex from $Y$. Since $|Y| = t+1$, at most $t+1$ agents can be assigned vertices only from $B$. Therefore, there is at least one Type-I agent $i$ such that $V(\pi_i) \subseteq B'$.

Now, consider $\pi_i$ ($i$ is a Type-I agent). We know that, due to Observation~\ref{O:min2} and Observation~\ref{O:PreEFO},  $u_i(\pi_i) \geq 5kt+t+1$. 
Now, $u_i(\pi_i) = u_i(\pi_i \cap (T \cup \{d_1\}))+ u_i(\pi_i\cap N) + u_i(\pi_i \cap W)$. Recall that for $v\in T\cup \set{d_1}$, $u_i(v) = 1$ and for $v' \in N$, $u_i(v') = 0$. Therefore, $u_i(\pi_i) \leq t+1+ 0 + u_i(\pi_i \cap W)$. We have already established that $u_i(\pi_i) \geq 5kt+t+1$. Therefore, $u_i(\pi_i \cap W) +t+1 \geq 5kt+t+1$, i.e., $u_i(\pi_i \cap W) \geq 5kt$. Since for each vertex $w\in W$, $u_i(w) = 5t$, we have that $|\pi_i \cap W| \geq k$. Due to Observation~\ref{O:notmany}, we have that $|\pi_i \cap W| \leq k$. Therefore, we have that $i$ is assigned exactly $k$ vertices from $W$. Using similar arguments, we have that $u_i(\pi_i)= u_i(\pi_i \cap (T\cup \{d_1 \}))+5kt \geq 5kt+t+1$, and hence, $u_i(\pi_i \cap (T\cup \{d_1 \})) \geq t+1$. Since $u_i(v) = 1$, for each $v\in T\cup \{d_1\}$, we have that $\pi_i \cap (T\cup \{d_1 \}) = T\cup \{d_1\}$. 

Let $W' = \pi_i \cap W$ and let $S \subseteq N$ be the set of vertices that are uniquely matched to $W'$. Observe that $S\subseteq \pi_i$. Note that $|S| =k$. We have the following claim to complete our proof.

\begin{clm}
$S$ is an RBDS of $G$.    
\end{clm}
\begin{proofofclaim}
We distinguish the following two cases.
\begin{enumerate}
\item $\pi_i \cap N = S$. Since $T \cup \{d_1\}$ is an independent set and vertices of $T \cup \{d_1\}$ can only have edges to vertices in $N$ in $G'[B']$, we have that $S$ dominates each vertex in $T\cup \{d_1\}$. Hence, $S$ dominates each vertex in $T$ and is an RBDS of $G$.

\item $\pi_i \cap N \neq S$. In this case, let $S' = \pi_i \cap N$. Observe that $S\subseteq S'$. Due to Observation~\ref{O:notmany}, we have that $|S'| = k+1$. Therefore, let $v$ be the vertex such that $S' = S \cup \{v\}$. Similarly to the previous case, we have that $S'$ is an RBDS of $G$. Moreover, similarly to the arguments used in the proof of Observation~\ref{O:notmany}, it is easy to see that no vertex $v' \in S$ is in $\tau_i$ (as $G'[\pi_i \setminus \{v'\}]$ contains at least two connected components, one containing the unique vertex $w' \in W$ attached to $v'$). 

Next, we establish that $v\in \tau_i$. Targeting contradiction assume that $v \notin \tau_i$. Hence, for each vertex $a \in S'$, $a\notin \tau_i$, and hence for any vertex $v\in \tau_i$, $\pi_i \setminus \{v\}$ contains at least $k+1$ vertices from $N$. This contradicts the fact that $\Pi$ is an \EFO (resp., \EFX) allocation due to Observation~\ref{O:notEFX}. Therefore, $v\in \tau_i$.

Now, if we can establish that for each vertex $v' \in N_T(v)$, $v' \in N(S)$, then this will prove that $S$ is an RBDS of $G$ (since $S'$ is an RBDS of $G$). Targeting contradiction assume that there is some $v' \in N_T(V)$ such that $v' \notin N(S)$. Then, since $T$ is an independent
set, $G'[\pi_i \setminus \{v\}]$ contains at least two connected components, one of which contains the vertex $v'$ (and another contains $d_1$). Therefore, $v \notin \tau_i$, which contradicts the fact that $v\in \tau_i$.  Hence, $S$ is an RBDS of $G$.
\end{enumerate}
This completes the proof of our claim. 
\end{proofofclaim}

Finally, since $S$ is an RBDS of $G$ and $|S|=k$, we have that $((G,k),t+k)$ is a Yes-instance of \RBDS.
\end{proof}



Finally, we have the following theorem as a consequence of Proposition~\ref{R:RBDS}, and Lemma~\ref{L:IPROP}, Lemma~\ref{L:IEFOdd}, Lemma~\ref{L:IEFEven}, Lemma~\ref{L:OneSide}, Lemma~\ref{L:EFOreverse}, and the fact in all our constructions in Section~\ref{S:NPoly}, the value of $\mathsf{vcn}+\mathsf{val}+|A|+p$ is bounded by a function polynomial in $k+|T|$.

\NoPoly*

\section{Color Coding For \PROP-\CFD}\label{S:ColorProp}
In this section, we present a color-coding~\cite{cc} based \FPT~algorithm for \PROP-\CFD~parameterized by $p$. An $(n,k)$-\textit{perfect hash family} $\mathcal{F}$ is a family of functions from $[n]$ to $[k]$ such that for every set $S\subseteq [n]$ of size $k$, there exists a function $f\in \mathcal{F}$ that is injective on $S$. We shall use the following result.

\begin{proposition}[\cite{splitter}]\label{P:cc}
For any $n,k\geq 1$, we can construct an $(n,k)$-perfect hash family of size $e^kk^{\mathcal{O}( \log k)}\log n$ and in  $e^kk^{\mathcal{O}( \log k)}$ $n\log n$ time.
\end{proposition}
    
\begin{definition}[Coloring]
For a graph $G$, a {\em coloring of $G$ using $k$ colors} is a function $col:V(G)\rightarrow[k]$. A connected subgraph $H$ of $G$ is {\em colorful} if all vertices in $H$ are colored with distinct colors.
\end{definition}


In the following lemma, we solve a sub-problem that we will use for our algorithm for \PROP-\CFD. 
\begin{lemma}\label{L:imp}
Let $G$ be a vertex-weighted graph with weight function $w:V(G)\rightarrow \mathbb{N}_0$. Then, given numbers $t,k \in \mathbb{N}$, we can decide in \FPT~time whether there exists a connected subgraph $H$ of $G$ such that $|V(H)|= k$ and $\sum_{v\in V(H)} w(v) \geq t$.
\end{lemma}
\begin{proof}
We will use color-coding to find such a subgraph. Let the vertices of $G$ be colored using $k$ colors. A coloring of the vertices of $G$ is a \textit{good coloring} if there exists a colorful subgraph $G'$ in the coloring that satisfies $|V(G')|= k$ and $\sum_{v\in V(G')}w(v) \geq t$. If we are dealing with a Yes-instance, then due to Proposition~\ref{P:cc}, we can get such a coloring in $e^kk^{\mathcal{O}(k \log k)}m\log m$ time. (In total $e^kk^{\mathcal{O}(k \log k)}\log m$ colorings will be generated and if we have a Yes-instance, at least one of these colorings will be a good coloring. )

Now, for each of the generated coloring ($e^kk^{\mathcal{O}(k \log k)}\log m$ colorings in total), we show how to find a colorful subgraph $H$ in $G$, if exists, such that $|V(H)|= k$ and $\sum_{v\in V(H)} w(v) \geq t$ using dynamic programming. For $v\in V(G)$ and $C\subseteq [k]$, let $M[v,C]$ be a table entry whose purpose is to hold the maximum weight of a subgraph, say, $H'$, such that each vertex in $V(H')$ has a distinct color from $C$ (all colors from $C$ are used) and $v\in V(H')$. Note that $|V(H')|= |C|$. 

Now, for each vertex $v\in V(G)$, let $M[v,col(v)] = w(v)$. For $|C|>2$, let $M[v,C] = \max_{u\in N[v], C'\subseteq C\setminus col(v)} M[u,C']+M[v,C\setminus C']$. Finally, if for some $v\in V(G)$, $M[v,[k]] \geq t$, then report a Yes-instance, else, report a No-instance. The correctness of this dynamic programming can be proved formally using standard induction techniques.

Note that our dynamic programming table has size $|V(G)|\cdot2^k$ and computing each entry takes polynomial time. Hence, our algorithm takes $2^k\cdot |V(G)|^{\mathcal{O}(1)}$ time for one fixed coloring.
Since we generate $e^kk^{\mathcal{O}(k \log k)}\log m$ colorings in total (using Proposition~\ref{P:cc}) in $e^kk^{\mathcal{O}(k \log k)}m\log m$ time, and then use dynamic programming on each coloring, we can solve our problem in $e^kk^{\mathcal{O}(k \log k)}2^km^{\mathcal{O}(1)}$ time in total.
  \end{proof}

Let $((G,A,\mathcal{U},p),p)$ be an instance of \PROP-\CFD~parameterized by $p$. Let the vertices of $G$ be colored using $n$ colors (since $n=|A|$, the set of colors used is same as $A$). A coloring is  \textit{suitable} if there exists a \PROP~allocation $\Pi = (\pi_1,\ldots, \pi_n)$ such that each vertex in $\pi_i$ is colored with the color $i$.  
Now, we have the following lemma.
\begin{lemma}\label{L:CC}
Given a coloring $col:V(G)\rightarrow A$  for an instance $(G,A,\mathcal{U},p)$ of \PROP-\CFD, we can decide whether this is a suitable coloring for $(G,A,\mathcal{U},p)$ in $e^pp^{\mathcal{O}(p \log p)}m ^{\mathcal{O}(1)}$ time.
\end{lemma}
\begin{proof}
For $i\in A$, Let $V_i = \set{v~|~col(v) = i}$ and let $G_i = G[V_i]$. Observe that $col$ is a suitable coloring if and only if there exists an allocation $\Pi=(\pi_1,\ldots,\pi_n)$ such that for each vertex $v\in \pi_i$, $col(v)=i$. Let $p_i=|\pi_i|$. Then, observe that $\sum_{i\in A} p_i=p$. Alternately, we can also say that $col$ is a suitable coloring if and only if there exists a vector $(p_1,\ldots,p_n)$ such that $p_i\in [p]$ and $\sum_{i\in A} p_i = p$, and for each $i\in A$, there exists a subset $\pi_i \subseteq V_i$ with the properties: (1) $G_i[\pi_i]$ is connected, (2)$|\pi_i|=p$, and (3) for $w=u_i$, $\sum_{v\in \pi_i} w(v) \geq \frac{u_i(V(G)}{n}$. 

Hence for each vector $(p_1,\ldots,p_n)$ such that $\sum_{i\in A} p_i  =p$, we use Lemma~\ref{L:imp} to decide if there exists a subset $\pi_i \subseteq V_i$ with the properties: (1) $G_i[\pi_i]$ is connected, (2)$|\pi_i|=p$, and (3) for $w=u_i$, $\sum_{v\in \pi_i} w(v) \geq \frac{u_i(V(G)}{n}$ in $e^pp^{\mathcal{O}(p \log p)}m^{\mathcal{O}(1)}$ time (since $p_i \leq p$). If for some vector $(p_1,\ldots,p_n)$ such that $\sum_{i\in A}p_i = p$, Lemma~\ref{L:imp} returns a Yes-instance, else, if for every such vector  Lemma~\ref{L:imp} returns a No-instance, then we declare a No-instance.

It is well known that the total number of vectors $(p_1,\ldots,p_n)$ such that $\sum_{i\in A}p_i = p$ is ${{p-1}\choose {n-1}}\leq p^n \leq p^p$ (since $n\leq p$). So, we may call Lemma~\ref{L:imp} at most $p^p$ times and each call takes $e^pp^{\mathcal{O}(p \log p)}m^{\mathcal{O}(1)}$ time. Hence, the algorithm is this lemma takes $p^pe^pp^{\mathcal{O}(p \log p)}m^{\mathcal{O}(1)}$ time.
  \end{proof}

Now, consider the instance $(G,A,\mathcal{U},p)$. We color the vertices of $G$ using $|A|$ colors uniformly at random. Next, we establish the probability of getting a suitable coloring of $G$ given that $(G,A,\mathcal{U},p)$ be a Yes-instance in the following lemma.
\begin{lemma}\label{L:color}
If $(G,A,\mathcal{U},p)$ is a Yes-instance of \PROP-\CFD, then coloring the vertices of $G$ uniformly at random using $|A|$ colors returns a suitable coloring with probability at least $(\frac{1}{|A|})^p$. 
\end{lemma}
\begin{proof}
Let $\set{v_1,\ldots,v_p} \subseteq V(G)$  be a set of vertices that is assigned to the agents in some (fixed) \PROP~allocation, say, $\Pi$, of $(G,A,\mathcal{U},p)$. Let the vertex $v_i$ is assigned to the agent $j$ in $\Pi$ ($i\in [p]$ and $j\in A$). The probability of $v_i$ getting color $j$ is $\frac{1}{|A|}$. Since each vertex is colored independently, the probability that each vertex $v_i$ gets color $j$, for $i\in [p]$ and $j\in A$, such that the vertex $v_i$ is assigned to the agent $j$ is $(\frac{1}{|A|})^p$. Since there might be multiple \PROP~allocations for $(G,A,\mathcal{U},p)$, the desired probability is at least $(\frac{1}{|A|})^p$.
  \end{proof}

Finally, we have the following theorem. 




\CCoding*
\begin{proof}
Using Lemma~\ref{L:CC} and \ref{L:color}, we have an algorithm that given a Yes-instance, returns a solution with probability at least $\frac{1}{n^p}$. Clearly, by repeating this algorithm $n^p$ times, we obtain the promised constant probability bound. More specifically, the probability of failure of our algorithm all these $n^p$ times is at most $(1-\frac{1}{n^p})^{n^p} \leq \frac{1}{e}$. Therefore, the success probability is at least $1-\frac{1}{e}$.

Since one execution of Lemma~\ref{L:CC} requires $e^pp^{\mathcal{O}(p \log p)}m^{\mathcal{O}(1)}$ time and we repeat it  $n^p \leq p^p$ times (since $n\leq p$), we have an overall running time of $e^pp^{2p+\mathcal{O}(p \log p)}m^{\mathcal{O}(1)} =$ $e^pp^{\mathcal{O}(p \log p)}m^{\mathcal{O}(1)}$. 
  \end{proof}

\section{Conclusion}
In this paper, we defined \CFD, a generalization of \CFDO, and conducted a comprehensive parameterized analysis of the problem. On the one hand, we proved that for $\varphi \in \{\mathsf{EF}, \mathsf{EF1},\mathsf{EFX}\}$, $\varphi$-\CFD remains W[1]-hard parameterized by $p+\mathsf{vcn}+|A|$, even when $\mathsf{vcn}+|A|$ is a small constant. On the other hand, we proved that \PROP-\CFD admits a randomized \FPT algorithm parameterized by $p$ alone. We also presented positive  as well as negative results  on the kernelization complexity of \CFD (for all fairness notions considered) parameterized by $\mathsf{vcn}+\mathsf{val}+p$.

There are several other generalizations of \CFD (or \CFDO) that can be practically well motivated as well as theoretically interesting. For instance, the requirement of connectivity can be   ``relaxed'' (for example, by requiring that each agent receives at most $c$ connected components where $c\geq 1$) as well as ``restricted'' (for example, by requiring that each agent should receive a connected component that is 2-connected, or, more generally, $c$-connected where $c\geq 1$). Moreover, it might be desirable sometimes that the diameter of each connected component (or the sum of diameters of all connected components) assigned to an agent is bounded by a parameter $d$. For the examples we have already described in the Introduction, it is easy to see that such variants make sense.

Furthermore, several other notions of \textit{fairness}, \textit{welfare}, and \textit{chores} were studied with respect to classical \textsc{Fair Division} (see, e.g., \cite{bouveret2016characterizing, caragiannis2019unreasonable,igarashi2019pareto}). Hence, it might be interesting to consider \CFD in light of these notions. Moreover, it might be interesting to study \CFD by considering the parameter $|V(G) - p|$ along with other parameters. This might be specifically applicable when the allocator wants to allocate as much resources as possible. 

In the classical fair division (\CFD on a clique with $p=|V(G)|$), it is known that \EFO always exists and can be computed in polynomial time using the famous envy-cycles  procedure~\cite{lipton2004approximately}. On the other hand,  the guarantee of the existence of \EFX exists for very restricted settings (for example, 4 agents with additive valuations~\cite{berger2022almost} and 2 agents with arbitrary valuations~\cite{plaut2020almost}). It will be interesting to study if this difference in the ``difficulty'' of the two notions extends to the setting of \CFD and/or to the setting of \CFDO.

Finally, we discuss a natural variant of \PROP-\CFD, namely $\mathsf{PRP}$-\CFD, where every agent measures its happiness not in the context of its valuation of the whole graph $G$, but the subgraph $G'$ of $G$ that is assigned to the agents. More formally, an allocation $\Pi = (\pi_1,\ldots, \pi_n)$ respects $\mathsf{PRP}$-\CFD if for every $i\in A$, $u_i(\pi_i) \geq \frac{u_i(V(G'))}{n}$, where $G' = G[\pi_1 \cup \cdots \cup \pi_n]$. Observe that, for $\mathsf{PRP}$-\CFD, RR~\ref{RR1} is safe, and is sufficient to get an exponential kernel parameterized $\mathsf{vcn}+\mathsf{val}+p$, unlike \PROP-\CFD, which requires more careful reduction rules. Moreover, the color-coding based \FPT algorithm for \PROP-\CFD is unlikely to work for $\mathsf{PRP}$-\CFD since the key ingredient of this algorithm is that we know beforehand the exact valuation that will make an agent ``happy'', which is not the case for $\mathsf{PRP}$-\CFD. In this regard, $\mathsf{PRP}$-\CFD behaves more like \EF, \EFO, and \EFX \CFD. Hence, it will be interesting to figure out whether $\mathsf{PRP}$-\CFD is \FPT or W[1]-hard when parameterized by $p+\mathsf{vcn}$.



 \bibliographystyle{plainurl} 

\bibliography{main}

\end{document}